\documentclass[a4paper]{article}

\usepackage{a4wide}
\usepackage{amsfonts}
\usepackage{graphicx}
\usepackage{epstopdf}
\usepackage{algorithm,algorithmic}
\usepackage{amsmath}
\usepackage[amsmath,thmmarks,hyperref]{ntheorem}
\usepackage{xcolor, hyperref}
\usepackage{booktabs}
\usepackage{mathrsfs}  
\usepackage{enumitem}
\setlist[enumerate]{leftmargin=.5in}
\setlist[itemize]{leftmargin=.5in}
\graphicspath{{Figures}{./}}

\newcommand{\proofbox}{{\rule{1.5ex}{1.5ex}}}
\theoremstyle{nonumberplain}
\theoremheaderfont{\normalfont\itshape}
\theorembodyfont{\normalfont}
\theoremseparator{.}
\theoremsymbol{\proofbox}
\newtheorem{proof}{Proof}

\theoremstyle{plain}
\theoremheaderfont{\normalfont\sffamily}
\theorembodyfont{\normalfont\itshape}
\theoremseparator{.}
\theoremsymbol{}
\newtheorem{theorem}{Theorem}

\newcommand{\newsiamthm}[2]{
  \theoremstyle{plain}
  \theoremheaderfont{\normalfont\sffamily}
  \theorembodyfont{\normalfont\itshape}
  \theoremseparator{.}
  \theoremsymbol{}
  \newtheorem{#1}[theorem]{#2}
}

\newsiamthm{lemma}{Lemma}
\newsiamthm{corollary}{Corollary}
\newsiamthm{proposition}{Proposition}
\newsiamthm{definition}{Definition}
\newtheorem{remark}[theorem]{Remark}

\title{Mean-variance dynamic portfolio allocation with transaction costs: a Wiener chaos expansion approach}

\newcommand{\email}[1]{#1}
\author{Areski Cousin\thanks{Institut de Recherche en Math\'ematique Avanc\'ee, Universit\'e de Strasbourg,
7 rue Ren\'e Descartes, 67084 Strasbourg, cedex, France. (\email{a.cousin@unistra.fr}).} 
\and J. Lelong\thanks{Univ. Grenoble Alpes, CNRS, Grenoble INP, LJK, 38000 Grenoble, France. 
  (\email{jerome.lelong@univ-grenoble-alpes.fr}).}
\and T. Picard\thanks{Univ. Grenoble Alpes, LJK, 38000 Grenoble, France and Nexialog, 75011 Paris, France.
  (\email{tom.picard@grenoble-inp.org}).}}

\newcommand{\Cov}{{\mathbb{C}\mbox{ov}}}
\newcommand{\Var}{{\mathbb{V}\mbox{ar}}}

\newcommand{\bfA}{\mbox{\boldmath $A$}}

\newcommand{\E}{{\mathbb E}}

\newcommand{\bfF}{\mbox{\boldmath $F$}}

\newcommand{\Prob}{{\mathbb P}}

\newcommand{\cF}{{\cal F}}
\newcommand{\cG}{{\cal G}}

\newcommand{\be}{\begin{equation}}
\newcommand{\ee}{\end{equation}}
\newcommand{\bea}{\begin{eqnarray}}
\newcommand{\eea}{\end{eqnarray}}
\newcommand{\beaa}{\begin{eqnarray*}}
\newcommand{\eeaa}{\end{eqnarray*}}

\def\cF{\mathcal{F}}
\def\cG{\mathcal{G}}

\def\cM{\mathcal{M}}

\def\cU{\mathcal{U}}

\def\bP{\mathbb{P}}
\def\bQ{\mathbb{Q}}
\def\bR{\mathbb{R}}

\DeclareMathOperator{\spane}{span}
\DeclareMathOperator{\sharpe}{Sharpe}

\DeclareMathOperator{\sign}{sign}
\DeclareMathOperator{\Diag}{Diag}

\definecolor{myblue}{RGB}{58, 110, 255}
\definecolor{blue}{HTML}{1F77B4}
\definecolor{orange}{HTML}{FF7F0E}
\definecolor{green}{HTML}{2CA02C}

\usepackage{dsfont}
\usepackage{hhline} 
\usepackage{scalerel,stackengine}
\usepackage[toc,page]{appendix}
\stackMath
\newcommand\reallywidehat[1]{%
\savestack{\tmpbox}{\stretchto{%
  \scaleto{%
    \scalerel*[\widthof{\ensuremath{#1}}]{\kern-.6pt\bigwedge\kern-.6pt}%
    {\rule[-\textheight/2]{1ex}{\textheight}}
  }{\textheight}%
}{0.5ex}}%
\stackon[1pt]{#1}{\tmpbox}%
}
\begin{document}

\maketitle

\begin{abstract}
This paper studies the  multi-period mean-variance portfolio allocation problem with transaction costs. Many methods have been proposed these last years to challenge the famous uni-period Markowitz strategy.
But these methods cannot integrate transaction costs or become computationally heavy and hardly applicable. In this paper, we try to tackle this allocation problem by proposing an innovative approach which relies on representing the set of admissible portfolios by a finite dimensional Wiener chaos expansion. This method is able to find an optimal strategy for the allocation problem subject to transaction costs. To complete the study, the link between optimal portfolios submitted to transaction costs and the underlying risk aversion is investigated. Then a competitive and compliant benchmark based on the sequential uni-period Markowitz strategy is built to highlight the efficiency of our approach. \\

\noindent{\bf Classification}: 62L20, 91G10, 91G60, 93E20
 
\end{abstract}

\section{Introduction}

Dynamic portfolio selection is one of the most studied topic in financial economics. The problem consists in allocating the wealth of an investor, among a basket of assets, over time. Finding the optimal portfolio is a difficult challenge since it depends on the objective of the investor. The Markowitz mean-variance formulation represents a first answer to this problem by providing fundamental basics for static portfolio allocation in a uni-period case (see \cite{markowits1952portfolio}). 
Mean-variance framework offers to build a portfolio of assets such that the expected return is maximized for a given level of risk. This method is easy to apply and has the favor of asset managers. Nevertheless, when the time horizon increases, this myopic strategy which cannot see ahead of the next time period, cannot challenge the dynamic optimal portfolio obtained from the multi-period version of the problem.  
\cite{merton1969lifetime} is one of the first paper studying multi-period portfolio investment in a dynamic programming framework. In this seminal paper, the authors consider a problem with one risky asset
and one risk-free asset. At each date, the investor can re-balance its wealth between the two assets, seeking to maximize an utility of the final time horizon wealth. They derive a simple closed-form expression for the optimal policy when there are no constraint or transaction cost. In a companion paper, \cite{samuelson1975lifetime} derives the discrete-time analog approach. The results presented in those studies and the innovative and promising aspect of the multi-period portfolio selection have stimulated the interest of the related scientific community. In the following years, the literature in multi-period portfolio selection has considerably grown, dominated by maximizing expected utility of terminal wealth of elementary forms as logarithm, exponential or CRRA functions. Dynamic programming techniques turn out to be the most suitable approaches to solve this kind of problems.
\\
However, important difficulties due the non separability of the problem in the sense of dynamic programming, have been reported in finding the optimal portfolio issued from the multi-period mean-variance approach. Nevertheless, \cite{li2000optimal} and \cite{zhou2000continuous} provide explicit formulation for the unconstrained multi-period mean–variance optimal portfolio both in a discrete and continuous time setting. \cite{li2002dynamic} derives the optimal portfolio policy for the continuous-time mean–variance model with no-shorting constraint. \cite{cui2014optimal} extends this work to provide a discrete framework. Even if these last studies have become increasingly realistic, the ignorance of transaction costs, hinders their efficient applications in real life. Transaction costs have a major impact on the optimal policy and cannot be ignored. 
\\
The integration of transaction costs has been widely studied in the uni-period mean-var case (see \cite{best2007algorithm}, \cite{lobo2007portfolio}, \cite{pogue1970extension}, \cite{xue2006mean}, \cite{yoshimoto1996mean}). In a continuous time setting, the problem is not recent, especially when the time horizon is considered infinite (see \cite{dai2008penalty}, \cite{davis1990portfolio}, \cite{dumas1991exact}, \cite{morton1995optimal}, \cite{muthuraman2006multidimensional}). In a mean-var and finite time horizon setting, \cite{dai2010continuous} studies the properties of the optimal strategies and boundaries which define the buy, sell and no trade-regions. Discrete time allocation strategies submitted to transaction costs have also been widely pursued. \cite{constantinides1979multiperiod} and \cite{holden2013optimal} investigate optimal investment policies with proportional costs, accompanied by fixed costs for the second. They also describe them in terms of a no-trade region in which it is optimal to leave the portfolio allocation unchanged. 
But many difficulties have been reported in the literature to design efficient and accurate methods to compute the optimal solutions. Furthermore, when solutions are proposed, they remain computationally heavy and hardly applicable. Some of them tractably solve the problem in several special cases (see \cite{boyd2013performance}, \cite{draviam2002generalized} and \cite{peng2011new}) or search for sub-optimal policies
(see \cite{bertsimas2008robust}, \cite{calafiore2008multi},   \cite{li2022multi} and \cite{topcu2008multistage}). 
\\
The most prevalent methods to tackle this kind of problem are based on stochastic control techniques. However, they also suffer from the same drawbacks. 
In order to limit the dimensionality, \cite{cai2013numerical} uses Chebyshev polynomials to interpolate the value functions on a sparse grid of the space. 
Many of other related studies such as \cite{gennotte1994investment}, \cite{steinbach1999markowitz} and \cite{wang2013multi} rely on trees for modeling rates of return. \cite{al2020multi} goes further by assuming that every future rates of return is known by the investor. Under the same assumption, \cite{brown2011dynamic} derives an upper bound to measure and highlight the good performances of heuristic strategies.  
\cite{boyd2013performance} performs an ADP (Approximate dynamic programming) by using sub-optimal solutions to approximate value functions. These sub-optimal solutions are issued from a quadratic version of the problem or from MPC method (model predictive control). \cite{cong2016multi} uses an other sub-optimal solution, called multi-stage strategy to tune the exploring phase of its backward recursion algorithm. This method provides a solution at least, as good as the sub-optimal one. Recently, \cite{pun2022optimal} has also adopted dynamic programming but is forced to handle different cases separately, which makes their method hardly computationally applicable.
\\
\\
In this paper, we attempt to fill this gap by
providing a new computational scheme to solve multi-periods portfolio allocation problem submitted to transaction costs. We address this problem by proposing an innovative approach which relies on representing the set of admissible portfolios by their finite dimensional Wiener chaos expansion. This numerical method estimates optimal portfolios submitted to proportional transaction costs. The policies are computed thanks to a stochastic gradient descent algorithm, and require no exploration framework, which was a major step of the methods based on stochastic control algorithms (see \cite{bachouch2021deep}). 
Then a competitive benchmark, based on the sequential uni-period Markowitz strategy is built to highlight the efficiency of our approach. This benchmark relies on the independence between the Sharpe ratio and the risk aversion. 
\\
The main contribution of this paper is threefold. We introduce an innovative and efficient numerical method to get optimal portfolios submitted to transaction costs. Then, we study the links between risk aversion and multi-period optimal portfolios submitted to transaction costs. Finally, we provide a reliable benchmark with sequential mean-variance uni-period models in the context of transaction costs mentioned above. To the best of our knowledge, our paper is the first to provide this kind of benchmark. 
\\
The remaining of this paper is organized as follows. Section \ref{sec:framework} is dedicated to describe the mean-variance problem for multi-period portfolio allocation submitted to transaction costs. Our methodology which aims at finding optimal portfolios in this context, is presented in Section \ref{sec:optimal-solution}.
In section \ref{sec:resolution}, the link between risk aversion and optimal solutions is investigated. 
Finally, in Section \ref{sec:result}, we show the efficiency of our solution and examine the impact of transaction cost by comparing performances of the presented models with benchmark models such as the sequential uni-period Markowitz approach described in Appendix \ref{sec:benchmark}. 

\section{Environment}
\label{sec:framework}
This section defines the environment in which we address the mean-variance allocation problem. In particular, we present the general mean-variance formulation for multi-period portfolio allocation with transaction costs. Then, we specify the dynamics of risky assets and reformulate the initial problem. 

\subsection{Mean-variance formulation for multi-period portfolio allocation}
We define the filtered probability space $(\Omega,\bfA,\Gamma = (\cG_t)_{t \in [0,T]},\mathbb{P})\,,$ with $\Gamma=\sigma(W)$, where $W=(W^1,\dots,W^{d})$ is a Brownian motion defined on $[0,T]$ with values in $\bR$.  For $N$ $\in \mathbb{N}$, $N>0$, we introduce the discrete-time grid $0=t_0<t_1<\dots<t_{N}=T$.  Let $ \bfF$, be the discrete time filtration generated by the Brownian increments on this grid, $\cF_n=\sigma(W_{t_k}, \ k\leq n)$ for $0 \le n \le N$. We denote as 
$\mathcal{L}=\bigcap_{p\geq 1}{L}^p(\Omega,\cF_N,\Prob)$. 
\\
\\
We consider a portfolio $V$, of initial wealth $v_0$, composed of the $d$ risky assets $S=(S^1,\dots,S^d)$ and the risk-free asset $S^0$. An investor, with a risk aversion $\gamma$ can re-allocate its portfolio $V$ at discrete times $t_1<\dots<t_{N}$. We define ($\alpha_n$, $\alpha_n^0)_{n\in\{1,N\}}$, the quantities of risky and risk free assets hold in the portfolio, such that $\forall n \in\{0,\dots,N\},\ V_n=\alpha_n \cdot S_n +\alpha_n^0S_n^0$. The processes $\alpha$ and $\alpha^0$ are assumed to be $\bf F$-predictable. The agent aims to find the strategy ($\alpha_n$, $\alpha_n^0)_{n\in\{1,N\}}$, such that the generated portfolio $V$, maximizes $\E_{\Prob}[V_N]-\gamma \E_{\Prob}\left[(V_N-\E_{\Prob}[V_N])^2\right]$.
\\
At time $0$, we assume that the owner has all its wealth in cash, ie $\alpha^0_0S_0^0=v_0$ and $\forall i\in\{1,\dots,d\}, \ \alpha^i_0=0 $. 
At time $n\in\{1,\dots,N-1\}$, the investor re-balances his portfolio from $(\alpha_n,\alpha_n^0)$ to $(\alpha_{n+1},\alpha_{n+1}^0)$. He pays the transaction costs, proportional to the trade cash volume per asset and equal to
$\sum_{i=1}^d \nu|\alpha_{n+1}^i-\alpha_{n}^i|S_n^i$.
The self-financing condition is 
$$
{V}_n=\alpha_{n+1}\cdot{S}_n +{\alpha}_{n+1}^0 S_n^0 + \sum_{i=1}^d\nu|\alpha_{n+1}^i-\alpha_{n}^i|S_n^i.
$$
The dynamic mean-variance portfolio allocation with transaction costs can be written as
\be
\label{sys:cost2}
\tag{$E^{\gamma}$}
\begin{aligned}
 & \sup_{(\alpha_n,\alpha_n^0)_n}
& & \E_{\Prob}\left[V_{N}-\gamma\left(V_N-\E_{\Prob}[V_N]\right)^2\right] \\
& \text{s.t.}
& & {V}_0=v_0,  \ (\alpha_n,\alpha_n^0)_n \ \bfF-Pred\\
& & & {V}_n=\alpha_nS_n +\alpha_n^0S_n^0=\alpha_{n+1}\cdot{S}_n +{\alpha}_{n+1}^0 S_n^0 + \sum_{i=1}^d\nu|\alpha_{n+1}^i-\alpha_{n}^i|S_n^i.
\\
\end{aligned}
\ee

\subsection{Framework and assumptions}
\label{sec:environment}
For any discrete time process $\Theta$, we write for $k,n 
\in \mathds{N}, \ k<n$ $\Theta_{k:n}=(\Theta_k,\Theta_{k+1},\dots,\Theta_{n-1},\Theta_n)$. Let $\Delta$ be the operator which associates its increments process $\Delta \Theta_n=\Theta_{n}-\Theta_{n-1}$.
We also define the normalized Brownian increments by 
$
\Delta \widehat{W_k}=\left(\frac{W^j_{t_{k}}-W^j_{t_{k-1}}}{\sqrt{t_{k}-t_{k-1}}}\right)_{j\in \{1,\ldots,d\}}$ for $1 \le k \le N$. 
\\
\\
We assume that the financial market they represent is complete. 
The risk free rate $(r_n)_{n\leq N}$ is assumed to be deterministic and $\forall n \in\{0,\dots,N\}$,  $ \log\left(\frac{S^0_{n+1}}{S^{0}_n}\right)=r_n(t_{n+1}-t_n)$. We assume that the risky assets are defined by
 \be
 \label{eq:asset}
\forall n\in\{1,\dots,N\}, \ i \in \{1,\dots,d\}, \,  \log\left(\frac{S_{n+1}^i}{S^{i}_n}\right)=\left(\mu_n^i-\frac{(\sigma_n^i)^T\sigma_n^i}{2}\right)(t_{n+1}-t_n)+ \sigma^i_n\cdot \Delta W_{n+1},
\ee
where $\mu$ and $\sigma$ are $\bfF$-adapted processes with values in $\mathbb{R}^d$ and $\mathbb{R}^{d\times d}$ respectively. The process $\sigma$ is called the volatility process and we assume that $\forall n \in\{1,\dots,N\},$ $\sigma_n$ is a.s. invertible, $\E_{\Prob}\left[|\sigma_n|^2\right]<\infty$
and that $S_n\in\mathcal{L}$. 
This model implies that knowing $\cF_n$, the returns $\left\{\frac{S_{n+1}^i}{S_n^i}, \ i\in\{1,\dots,d\}\right\}$ are independent of $\cF_n$. Note that local volatility models fit in this framework by considering the Euler scheme of $\log(S)$. 
Let $\bfF^S = (\cF_n^S)_{0 \le n \le N}$ be the natural filtration of the risky assets, $\cF_n^S=\sigma(S_{t_k}, \ k\leq n)$ . 
Since $(\mu^i)_{i\in\{1,\dots,d\}}$ and $(\sigma^i)_{i\in\{1,\dots,d\}}$ are $\bfF$-adapted processes and $\sigma_n$ is a.s. invertible, we can easily prove by induction that $\bfF=\bfF^S$. 
\\
\\
For $0 \le n\leq N$ we define $\Phi_n=\sigma_n^{-1}(\mu_n-r_n)$. Let  $\tau(t)=\sup\{t_n < t \,:\, n\leq N\}$. We assume that  $\E_{\Prob}\left[\exp\left(\int_{0}^T\frac{1}{2}|\Phi_{\tau(u)}|^2du\right)\right]<\infty$ such that the process
$\left(\exp\left(-\int_{0}^t\Phi_{\tau(u)}\cdot dW_u-\frac{1}{2}\int_{0}^t|\Phi_{\tau(u)}|^2du\right)\right)_{0 \le t \le T}$ is a martingale. Therefore, we know from Girsanov's theorem that there exists a probability measure $\mathbb{Q}$ equivalent to $\bP$ , defined by 
$$\frac{d\mathbb{Q}}{d\mathbb{P}}=\exp\left(-\int_{0}^T\Phi_{\tau(u)}\cdot dW_u-\frac{1}{2}\int_{0}^T|\Phi_{\tau(u)}|^2du\right).
$$
We denote as $\mathscr{Z}=\left(\exp\left(-\int_{0}^{t_n}\Phi_{\tau(u)}\cdot dW_u-\frac{1}{2}\int_{0}^{t_n}|\Phi_{\tau(u)}|^2du\right)\right)_{0\leq n\leq N}.$ 
Let $(W_t^\bQ)_{t \in [0,T]}$ be a $\Gamma$-adapted process defined by $W_0^\bQ = 0$ and
$$
dW_{t}^\bQ=dW_t+\Phi_{\tau(t)}dt,
$$
then $W^\bQ$ is a Brownian motion under $\bQ$. 
Note that
\be
 \label{eq:changement}
 \forall i\in\{1,\dots,d\},\ \Delta W_{n+1}^{\bQ,i}\stackrel{\Prob}{=}\Delta W_{n+1}^{i}+ \Phi^i_n(t_{n+1}-t_n).
 \ee
Let $ \bfF^\bQ = (\cF^\bQ_n)_{0 \le n \le N} $, be the discrete time filtration generated by $W^\bQ$, $\cF_n^\bQ=\sigma(W_{t_k}^\bQ, \ k\leq n)$.  
We notice that 
\begin{equation*}
\forall n\in\{1,\dots, N\}, \ i \in \{1,\dots,d\}, \  \log\left(\frac{S_{n+1}^i}{S^{i}_n}\right)=r_n(t_{n+1}-t_n)+ \sigma^i_n\cdot \Delta W_{n+1}^\bQ,
\end{equation*}
We can see by induction that we also have $\bfF^\bQ=\bfF^S$.
\\
We use the tilde notation to denote discounting and we notice that $\tilde{S}^i = S^i/S^0$ is a $\bfF$-martingale under the probability $\mathbb{Q}$. Lastly, we assume that the process $\E_\bQ\left[\Delta \tilde{S}_{n+1} \Delta \tilde{S}_{n+1}^T|\cF_{n}\right]$ is a.s. invertible and 
$\E_{\bQ}\left[\Delta \tilde{S}_{n+1}\Delta \tilde{S}_{n+1}^T|\cF_{n}\right]^{-1}, \ \mathscr{Z}_n \in\mathcal{L}$.

\subsection{Reformulation of the problem}
With the notations and assumptions of Section\ref{sec:environment}, the self-financing condition can be reformulated as
\be
\label{eq:dyn-cout}
\Delta \tilde{V}_{n+1}=\alpha_{n+1}\cdot\Delta \tilde{S}_{n+1}-\sum_{i=1}^d\nu|\alpha_{n+1}^i-\alpha_{n}^i|\tilde{S}_n^i.
\ee
We define the cumulative cost process $\mathcal{C}$ by
\be
\forall n \in \{1,\dots,N\}, \ \mathcal{C}_n=\sum_{k=0}^{n-1}\sum_{i=1}^d\nu|\alpha_{k+1}^i-\alpha_{k}^i|\tilde{S}_k^i, \ \mathcal{C}_0=0.
\ee
We decide to work with the $\bfF$-martingale under $\mathbb{Q}$, $\tilde{X}_n= \tilde{V}_n + \mathcal{C}_n$.
Let $\mathcal{M}$ be the space of squared integrable $\bfF$-martingales under $\mathbb{Q}$ and $\mathcal{M}_S$ the sub-space of $\mathcal{M}$ defined by 
$$
 \mathcal{M}_S=\left\{
    \begin{array}{ll}
       M\in\mathcal{M}: \ \exists \ (\alpha_{k})_{1\leq k\leq N} \ \bfF-\text{predictable} \ s.t \ \forall n\in\{0,\dots,N\}, \\
       M_n=h+\sum_{k=0}^{n-1}\alpha_{k+1}\cdot\Delta \tilde{S}_{k+1} 
    \end{array}
\right \}. 
$$
$\mathcal{M}_S$ is the set of martingales which are a martingale transformations of $\tilde{S}$. The space $\mathcal{M}_S$ also represents the space of admissible portfolios in discrete-time.
The mean-variance multi-period portfolio allocation problem with transaction costs \eqref{sys:cost2}, can be reformulated as
\be
\label{system_martingale2}
\tag{$E_{\mathcal{M}_S}^{\gamma}$}
\begin{aligned}
& \sup_{\tilde{X} \in \mathcal{M}_S}
& & \E_{\Prob}\left[(\tilde{X}_{N}-{\mathcal{C}}_N)S_N^0-\gamma\left((\tilde{X}_{N}-{\mathcal{C}}_N)S_N^0-\E_{\Prob}[(\tilde{X}_{N}-{\mathcal{C}}_N)S_N^0]\right)^2\right] \\
& \text{s.t.}
& & \tilde{X}_0=v_0,
\\
& & &  {\mathcal{C}}_{n+1}={\mathcal{C}}_n+\sum_{i=1}^d \nu|\alpha_{n+1}^i-\alpha_{n}^i|\tilde{S}_n^i\\
\end{aligned}
\ee
In the next section, we present an innovative approach to solve the dynamic mean-var allocation problem in presence of transaction costs.

\section{Main results}
\label{sec:optimal-solution}
The main contribution of this paper is to propose and study a numerically tractable approximation of~\eqref{system_martingale2}. First, we embed the original problem into a more standard stochastic optimization framework~\eqref{system_chaosT2}. Then, we use Wiener chaos polynomials to obtain a finite dimensional formulation~\eqref{system_chaos_K}.

\subsection{Embedding representation}
The main difficulty remains to parameterize $\mathcal{M}_S$. We propose to find optimal portfolios on $\mathcal{M}_S$ by exploring ${\mathcal{M}}$. Furthermore, the objective function $(\tilde{X}_{N}-{\mathcal{C}}_N)S_N^0-\gamma\left((\tilde{X}_{N}-{\mathcal{C}}_N)S_N^0-\E_{\Prob}[(\tilde{X}_{N}-{\mathcal{C}}_N)S_N^0]\right)^2$ depends on the law of the portfolio value. This non-separable form is not easy to manipulate.  
In order to solve \eqref{system_martingale2}, we would like to embed it into a tractable equivalent one.
\begin{proposition}
\label{prop:main}
The problem \eqref{system_martingale2} is equivalent to 

\be
\label{system_chaosT2}
\tag{$\mathcal{E}_\mathcal{M}^{\gamma}$}
\begin{aligned}
& \sup_{Z \in {\mathcal{M}}, \ \theta\in\mathbb{R}}
& & \E_{\Prob}\left[F^{\gamma}(Z,\theta)\right] \\
& \text{s.t.}
& & Z_0=v_0,
 \\
& & &  \mathcal{C}_{n+1}(Z)={\mathcal{C}}_n(Z)+\sum_{i=1}^d \nu|\alpha_{n+1}^i(Z)-\alpha_{n}^i(Z)|\tilde{S}_n^i\\
\end{aligned}
\ee
with 
\be
F^{\gamma}(Z,\theta)=\mathcal{R}(Z)S_N^0-\gamma\left((\mathcal{R}(Z)-\theta)S_N^0\right)^2,
\ee
\be
\mathcal{R}(Z)=\Pr\left(Z\right)_N-\mathcal{C}_N(Z),
\ee
\be
\label{eq:projection}
\Pr(Z) = \left(Z_0+\sum_{k=0}^{n-1}\alpha_{k+1}(Z)\cdot\Delta \tilde{S}_{k+1}  \right)_{0\leq n\leq N},
\ee 
\begin{equation}
\label{eq:control_X}
\alpha_{n+1}(Z) = \left(\E_\bQ\left[\Delta \tilde{S}_{n+1} \Delta \tilde{S}_{n+1}^T|\cF_{n}\right]\right)^{-1} \E_\bQ\left[\Delta Z_{n+1}\Delta \tilde{S}_{n+1}|\cF_{n}\right].
\end{equation} 

\end{proposition}
The proof of Proposition \ref{prop:main} relies on the following Lemma. 
\begin{lemma}
\label{prop:equivalent}
The problem \eqref{system_martingale2} is equivalent to
\be
\label{mathcal_E_gamma}
\tag{$\mathcal{E}^{\gamma}_{\mathcal{M}_S}$}
\begin{aligned}
& \sup_{\tilde{X} \in \mathcal{M}_S, \ \theta\in\mathbb{R}}
& & \E_{\Prob}\left[(\tilde{X}_{N}-{\mathcal{C}}_N)S_N^0-\gamma\left((\tilde{X}_{N}-{\mathcal{C}}_N-\theta)S_N^0\right)^2\right] \\
& \text{s.t.}
& & \tilde{X}_0=v_0,
\\
& & &  {\mathcal{C}}_{n+1}={\mathcal{C}}_n+\sum_{i=1}^d \nu|\alpha_{n+1}^i-\alpha_{n}^i|\tilde{S}_n^i\\
\end{aligned}
\ee
\end{lemma}
\begin{proof}[of Lemma \ref{prop:equivalent}]
Let $(T_n)_n$ be a $\bfF$-adapted process. We define the function
\beaa
\begin{aligned}
\cU\ :\ & \mathbb{R}\rightarrow \mathbb{R}\\
&\theta\longmapsto\E_{\Prob}[T_N]-\gamma\E_{\Prob}\left[\left((T_N-\theta)S_N^0\right)^2\right]
\end{aligned}
\eeaa
The function $\cU$ is a second order polynomial such that $\lim_{|\theta| \to \infty} \cU(\theta) = -\infty$. Then, $\cU$ attains its maximum at a unique point $\theta^\star$ defined by $\nabla \cU(\theta^{\star})=0$, $\theta^{\star}=\E_{\Prob}[T_N]$.
\end{proof}

\begin{proof}[of Proposition \ref{prop:main}]
 Let $M \in \mathcal{M}$. 
The projection of $M$ on $\cM_S$ writes as $$\left(h+\sum_{k=0}^{n-1} \alpha_{k+1}\cdot\Delta \tilde{S}_{k+1}\right)_{0\leq n\leq N},$$ with $\alpha_{k}^j$ $\bfF$-predictable, $h \in \mathbb{R}$. Let $B\in\mathcal{M}_S$ such that for $n \in\{0,\dots,N\}$, \ 
$B_n= g+\sum_{k=0}^{n-1} c_{k+1}\cdot\Delta \tilde{S}_{k+1}, \ g\in\mathbb{R}, \ c_{k+1}^j \ \bfF-predictable$. Then we have for $n\in\{1,\dots,N\}$,
\be
\E_\bQ\left[\left(M_n-h-\sum_{k=0}^{n-1} \alpha_{k+1}\cdot\Delta \tilde{S}_{k+1} \right)\left(g+\sum_{k=0}^{n-1} c_{k+1}\cdot\Delta \tilde{S}_{k+1}\right)\right]=0.
\ee
By taking $c_{k+1}^j=0$, for all $j \in\{1,\dots,d\}$, we get 
$h=M_0$. Setting $g=0$, we have
\be
\label{eq=0}
\E_\bQ\left[\left(\sum_{l=0}^{n-1}\Delta M_{l+1}-\sum_{k=0}^{n-1} \alpha_{k+1}\cdot\Delta \tilde{S}_{k+1} \right)\left(\sum_{k=0}^{n-1} c_{k+1}\cdot\Delta \tilde{S}_{k+1}\right)\right]=0.
\ee
By developing the expression and using that $M$ and $\tilde{S}^j$ are $\bfF$-martingale, we obtain for $l\neq k$, 
\be
\E_\bQ\left[ c_{k+1}^j \Delta M_{l+1} \Delta \tilde{S}_{k+1}^j \right]=\E_\bQ\left[ c_{k+1}^j \E_\bQ\left[\Delta M_{l+1} \Delta \tilde{S}_{k+1}^j|\cF_{l\vee k}\right] \right]=0.
\ee
Then, for $1\leq k\leq n-1$,
\be
\label{eqDelta}
\E_\bQ\left[\sum_{l=0}^{n-1} \Delta M_{l+1} c_{k+1}^j\Delta \tilde{S}^j_{k+1}\right]=\E_\bQ\left[c_{k+1}^j\Delta M_{k+1} \Delta \tilde{S}^j_{k+1}\right].
\ee
Finally, inserting (\ref{eqDelta}) in (\ref{eq=0}) yields 
\be
\E_\bQ\left[\sum_{k=0}^{n-1}  c_{k+1}\cdot\Delta \tilde{S}_{k+1} \Delta M_{k+1}-\sum_{k_1=0}^{n-1}\sum_{k_2=0}^{n-1} (\alpha_{k_1+1}\cdot\Delta \tilde{S}_{k_1+1}) (c_{k_2+1}\cdot \Delta \tilde{S}_{k_2+1})  \right]=0
\ee
In the same way, for $k_1 \neq k_2$, 
\be
\E_\bQ\left[ (\alpha_{k_1+1}\cdot\tilde{S}_{k_1+1})( c_{k_2+1}\cdot\Delta \tilde{S}_{k_2+1})  \right]=\E_\bQ\left[\E\left[ (\alpha_{k_1+1}\cdot\Delta\tilde{S}_{k_1+1}) (c_{k_2+1}\cdot\Delta\tilde{S}_{k_2+1})|\cF_{k_1 \vee k_2} \right] \right]=0.
\ee
Then 
\be
\sum_{k=0}^{n-1}\E_\bQ\left[ c_{k+1}\cdot\E_\bQ\left[\Delta \tilde{S}_{k+1} \Delta M_{k+1} |\cF_{k}\right]-  c_{k+1}\cdot\left(\E_\bQ\left[\Delta \tilde{S}_{k+1} \Delta \tilde{S}_{k+1}^T|\cF_k \right]\alpha_{k+1}\right)   \right]=0.
\ee
So
\be
\sum_{k=0}^{n-1}\E_\bQ\left[ c_{k+1}\cdot\left(\E_\bQ\left[\Delta \tilde{S}_{k+1} \Delta M_{k+1} |\cF_{k}\right]- \E_\bQ\left[\Delta \tilde{S}_{k+1} \Delta \tilde{S}_{k+1}^T|\cF_k \right]\alpha_{k+1}  \right) \right]=0.
\ee
As this equation holds for all $(c_k)_{1 \le k \le N}$, we deduce that $(\alpha_n)_{n}$ solves 
\begin{equation*}
\E_\bQ\left[\Delta M_{n+1}\Delta \tilde{S}_{n+1}|\cF_{n}\right]=\E_\bQ\left[\Delta \tilde{S}_{n+1} \Delta \tilde{S}_{n+1}^T|\cF_{n}\right] \alpha_{n+1}.
\end{equation*}
Then, we conclude that $(\alpha_n)_n$ is defined by~\eqref{eq:control_X} and that \eqref{mathcal_E_gamma} can be reformulated as~\eqref{system_chaosT2}.
\end{proof}

\subsection{Wiener chaos parametrization}
\label{sec:Wiener chaos expansion of the solution}
We will parametrize the elements of $\mathcal{M}$ through the Wiener chaos expansion of their terminal value. The basic theory of Wiener chaos expansion and the properties used here are presented in Appendix \ref{sec:chaos}.
 According to \cite[Theorem 2.1]{akahori2017discrete}, every element $Y$ of $L^2(\Omega,\cF_N,\mathbb{Q})$ can be represented by its Wiener chaos expansion as
\begin{equation}
Y = \E_\bQ[Y]+\sum_{\substack{\lambda \in (\mathbb{N}^{{N}})^d}} \beta_{\lambda} H_{\lambda}^{\bigotimes}(\Delta \widehat{W}^\bQ),
\end{equation}
where $$H_{\lambda}^{\bigotimes}(\Delta \widehat{W}^\bQ) = \prod_{j=1}^d\prod_{i= 1}^N H_{\lambda_{i}^j}\left(\frac{W^{\bQ,j}_{t_{i}}-W^{\bQ,j}_{t_{i-1}}}{\sqrt{t_{i}-t_{i-1}}}\right).$$
We define the truncated expansion of $Y$ of order $K$ under $\mathbb{Q}$ by $\mathscr{C}_K(Y)$ 
such that 
$$
\mathscr{C}_K(Y)=\E_\bQ[Y]+\sum_{\substack{\lambda \in (\mathbb{N}^{N})^d \\ |\lambda|_1\leq K}} \beta_{\lambda} H_{\lambda}^{\bigotimes}(\Delta \widehat{W}^\bQ).
$$
It is well-known that $\lim_{K \to +\infty} \E_\bQ\left[|Y-\mathscr{C}_K(Y)|^2\right]=0.$ 
Proposition \ref{prop:tronc} states that for $n\leq N$, $\E_{\bQ}\left[\mathscr{C}_K(Y)|\cF_{n}\right] = \mathscr{C}_K(\E_{\bQ}[Y|\cF_n])$, which can be obtained by removing the non $\cF_n$-measurable terms from the chaos expansion of $Y$. By identifying the martingale $Z\in\mathcal{M}$ with the chaos expansion of its terminal value $Z_N$,  
we slightly abuse the notation $\mathscr{C}_K$ 
to define the process $\mathscr{C}_K(Z)=(\mathscr{C}_K(Z_n))_n$.
\\
Let's consider the analog problem for $K>0$
\be
\label{system_chaos_K}
\tag{$\mathcal{E}_K^{\gamma}$}
\begin{aligned}
& \sup_{Z \in {\mathcal{M}}, \ \theta\in\mathbb{R}}
& & \E_{\Prob}\left[F_{K}^{\gamma}(Z,\theta)\right] \\
& \text{s.t.}
& & Z_0=v_0, \\
& & &  \mathcal{C}_{n+1}(\mathscr{C}_K(Z))={\mathcal{C}}_n(\mathscr{C}_K(Z))+\sum_{i=1}^d \nu|\alpha_{n+1}^i(\mathscr{C}_K(Z))-\alpha_{n}^i(\mathscr{C}_K(Z))|\tilde{S}_n^i
\end{aligned}
\ee
\beaa
\begin{aligned}
\text{with} \ F_{K}^{\gamma}(Z,\theta) =F^{\gamma}(\mathscr{C}_K(Z),\theta)=\mathcal{R}(\mathscr{C}_K(Z))S_N^0-\gamma\left((\mathcal{R}(\mathscr{C}_K(Z))-\theta)S_N^0\right)^2.
\end{aligned}
\eeaa 
The following proposition states that the optimum of \eqref{system_chaos_K} converges to the optimum of \eqref{system_chaosT2}. As a result, we can approach the optimum of  \eqref{system_chaosT2} by solving \eqref{system_chaos_K}. 
\begin{proposition}
\label{prop:convergence}
Let $(Z^{\star},\theta^{\star})$ be a solution to \eqref{system_chaosT2}. If $\forall \mathcal{X} \in \mathcal{L}, \ \sup_{K>1}\left[\mathcal{X}\mathscr{C}_K(Z^{\star})^2\right]<\infty$, then, there exists $\eta>0$ such that for all $K>0$ and $(Z^K,\theta^K)$ solution to \eqref{system_chaos_K}, we have 
$$
|\E_{\Prob}\left[F^{\gamma}(Z^{\star},\theta^{\star})\right]-\E_{\Prob}\left[F^{\gamma}_K(Z^K,\theta^K)\right]|\leq\eta\E_{\Prob}\left[|Z^{\star}_N-\mathscr{C}_K(Z^{\star}_N)|^2\right]^{\frac{1}{2}}.
$$
\end{proposition}
The proof of the proposition is based on the following lemmas.
\begin{lemma}
\label{lem:convergence}
Let $(Z^{\star},\theta^{\star})$ be a solution to \eqref{system_chaosT2}. If $\forall \mathcal{X} \in \mathcal{L}, \ \sup_{K>1}\left[\mathcal{X}   \
\mathscr{C}_K(Z^{\star})^2\right]<\infty$, then there exists $\eta>0$ such that 
$$
\left|\E_{\Prob}\left[F^{\gamma}(Z^{\star},\theta^{\star})\right]-\E_{\Prob}\left[F^{\gamma}_K(Z^{\star},\theta^{\star})\right]\right|\leq \eta\E_{\Prob}\left[|Z^{\star}_N-\mathscr{C}_K(Z^{\star}_N)|^2\right]^{\frac{1}{2}}.
$$
\end{lemma}
\begin{lemma}
Let $ 1 \leq p < 2$, then $\exists \ \zeta>0$, such that $\forall \ Z,\ Z^{'}\in \mathcal{M}$, $\forall n\in \{0,\ldots,N-1\}$, 
\be
\label{eq:alpha}
\E_{\Prob}\left[|\alpha_{n+1}(Z)-\alpha_{n+1}(Z^{'})|^p\right]\leq \zeta \E_{\bQ}\left[|\Delta Z_{n+1}-\Delta Z^{'}_{n+1}|^2\right]^{\frac{p}{2}},  
\ee
\be
\label{eq:Pr}
\E_{\Prob}\left[|Pr(Z^{\star})_N-Pr(\mathscr{C}(Z^{\star}))_N|\right]\leq \zeta \E_{\bQ}\left[\left| Z_{N}- Z^{'}_{N}\right|^2\right]^{\frac{1}{2}}.
\ee
\label{lem:alpha}
\end{lemma}

\begin{proof}[of Lemma \ref{lem:alpha}]
Using consecutively Jensen's inequality and Holder's inequality with the coefficients $\frac{2}{2-p}$ and $\frac{2}{p}$, we have 
\beaa
\begin{aligned}
  &\E_{\Prob}\left[|\alpha_{n+1}(Z)-\alpha_{n+1}(Z^{'})|^p\right] =  \E_{\Prob}\left[\left|\E_{\bQ}\left[\Delta \tilde{S}_{n+1}\Delta \tilde{S}_{n+1}^T|\cF_{n}\right]^{-1}\E_{\bQ}\left[\left(\Delta Z_{n+1}-\Delta Z_{n+1}^{'}\right)\Delta \tilde{S}_{n+1}|\cF_n\right]\right|^p\right]\\
  & \leq \E_{\bQ}\left[\E_{\bQ}\left[\left|\mathscr{Z}_n\right|\left|\left(\Delta Z_{n+1}-\Delta Z_{n+1}^{'}\right)\E_{\bQ}\left[\Delta \tilde{S}_{n+1}\Delta \tilde{S}_{n+1}^T|\cF_{n}\right]^{-1}\Delta \tilde{S}_{n+1}\right|^{p}|\cF_n\right]\right]
  \\
  & \leq \E_{\bQ}\left[\left|\mathscr{Z}_n^{\frac{1}{p}}\E_{\bQ}\left[\Delta \tilde{S}_{n+1}\Delta \tilde{S}_{n+1}^T|\cF_{n}\right]^{-1}\Delta \tilde{S}_{n+1}\right|^{\frac{2p}{2-p}}\right]^{\frac{2-p}{2}}\E_{\bQ}\left[\left|\Delta Z_{n+1}-\Delta Z_{n+1}^{'}\right|^2\right]^{\frac{p}{2}}.
\end{aligned}
\eeaa
By recalling that $\tilde{S}_{n+1}, \ \E_{\bQ}\left[\Delta \tilde{S}_{n+1}\Delta \tilde{S}_{n+1}^T|\cF_{n}\right]^{-1}, \ \mathscr{Z}_n \in\mathcal{L}$, we find $\zeta>0$ such that 
\\ 
$\E_{\bQ}\left[\left|\mathscr{Z}_n^{\frac{1}{p}}\E_{\bQ}\left[\Delta \tilde{S}_{n+1}\Delta \tilde{S}_{n+1}^T|\cF_{n}\right]^{-1}\Delta \tilde{S}_{n+1}\right|^{\frac{2p}{2-p}}\right]^{\frac{2-p}{2}}<\zeta$, and obtain (\ref{eq:alpha}). 
For the second inequality (\ref{eq:Pr}), we use Holder's inequality with the coefficients $p=\frac{3}{2}$, $q=3$, and (\ref{eq:alpha}). The constant $\zeta'>0$ may change over the inequalities.
\beaa
\begin{aligned}
& \E_{\Prob}\left[|Pr(Z^{\star})_N-Pr(\mathscr{C}(Z^{\star}))_N|\right]\leq \sum_{k=0}^{N-1}\E_{\Prob}\left[\left|\alpha_{k+1}(Z^{\star})\cdot\Delta \tilde{S}_{k+1}-\alpha_{k+1}(\mathscr{C}_K(Z^{\star}))\cdot\Delta \tilde{S}_{k+1}\right|\right]\\
& \leq \sum_{k=0}^{N-1}\E_{\Prob}\left[\left|\alpha_{k+1}(Z^{\star})-\alpha_{k+1}(\mathscr{C}_K(Z^{\star}))\right|^p\right]^{\frac{1}{p}}\E_{\Prob}\left[\left|\Delta \tilde{S}_{k+1}\right|^q\right]^{\frac{1}{q}} \leq \sum_{k=0}^{N-1 }\zeta^{'} \E_{\bQ}\left[\left|\Delta Z_{k+1}-\Delta\mathscr{C}_K(Z^{\star})_{k+1}\right|^2\right]^{\frac{1}{2}}.\\
& \leq \zeta^{'}\left(\sum_{k=0}^{N-1 } \E_{\bQ}\left[\left|\Delta Z_{k+1}-\Delta\mathscr{C}_K(Z^{\star})_{k+1}\right|^2\right]\right)^{\frac{1}{2}}=\zeta^{'}\E_{\bQ}\left[\left| Z_{N}-\mathscr{C}_K(Z^{\star})_{N}\right|^2\right]^{\frac{1}{2}}.
\end{aligned}
\eeaa
The last equality stems from the martingale property of $Z^{\star}-\mathscr{C}_K(Z^{\star})$. 
\end{proof}
\begin{proof}[of Lemma \ref{lem:convergence}]
We have 
\beaa
\label{eq:utility}
\begin{aligned}
\E_{\Prob}\left[|F^{\gamma}(Z^{\star},\theta^{\star})-F^{\gamma}_K(Z^{\star},\theta^{\star})|\right]  
& \leq \E_{\Prob}\left[|\mathcal{R}(Z^{\star})-\mathcal{R}(\mathscr{C}_K(Z^{\star}))|\right] S_N^0\\
& + \gamma\E_{\Prob}\left[|\left(\mathcal{R}(\mathscr{C}_K(Z^{\star}))-\theta^{\star}\right)^2-(\mathcal{R}(Z^{\star})-\theta^{\star})^2|\right](S_N^0)^2.
\end{aligned}
\eeaa
For the first term, we have 
\beaa
\begin{aligned}
& \E_{\Prob}\left[|Pr(Z^{\star})_N-\mathcal{C}_N(Z^{\star})-\left(Pr(\mathscr{C}_K(Z^{\star}))_N-\mathcal{C}_N(\mathscr{C}_K(Z^{\star}))\right)|\right] \\
& \leq \E_{\Prob}\left[|Pr(Z^{\star})_N-Pr(\mathscr{C}_K(Z^{\star}))_N|\right]  +
\sum_{i=1}^d \sum_{k=0}^{N-1}\nu \E_{\Prob}\left[|\alpha_{k+1}^i(\mathscr{C}_K(Z^{\star}))-\alpha_{k}^i(\mathscr{C}_K(Z^{\star}))-(\alpha_{k+1}^i(Z^{\star})-\alpha_{k}^i(Z^{\star}))||\tilde{S}^i_k|\right]\\
\end{aligned}
\eeaa
Using Lemma \ref{lem:alpha} and the same arguments as to prove (\ref{eq:Pr}), we find $\eta_1>0$, such that
\be
\label{eq:first-term}
\E_{\Prob}\left[|Pr(Z^{\star})_N-\mathcal{C}_N(Z^{\star})-\left(Pr(\mathscr{C}_K(Z^{\star}))_N-\mathcal{C}_N(\mathscr{C}_K(Z^{\star}))\right)|\right]\leq \eta_1\E_{\bQ}\left[\left|Z_N-\mathscr{C}_K(Z^{\star}_N)\right|^2\right]^{\frac{1}{2}}.
\ee
For the second term, we rewrite 
\beaa
\begin{aligned}
 & \E_{\Prob}\left[|\left(\mathcal{R}(\mathscr{C}_K(Z^{\star}))-\theta^{\star}\right)^2-(\mathcal{R}(Z^{\star})-\theta^{\star})^2|\right]=\E_{\Prob}\left[|\mathcal{R}(\mathscr{C}_K(Z^{\star}))+\mathcal{R}(Z^{\star})-2\theta^{\star}||\mathcal{R}(\mathscr{C}_K(Z^{\star}))-\mathcal{R}(Z^{\star})|\right].\\
 & \leq \E_{\Prob}\left[|\mathcal{R}(\mathscr{C}_K(Z^{\star}))+\mathcal{R}(Z^{\star})||\mathcal{R}(\mathscr{C}_K(Z^{\star}))-\mathcal{R}(Z^{\star})|\right]+ 2|\theta^{\star}|\E_{\Prob}\left[|\mathcal{R}(\mathscr{C}_K(Z^{\star}))-\mathcal{R}(Z^{\star})|\right]
\end{aligned}
\eeaa
Let $ i,j\in\{1,\ldots,d\}, \ k, l \in \{0,\ldots,N-1\}$ and $\mathcal{Y}\in\{\Delta \tilde{S}_{k+1}^i\Delta\tilde{S}_{l+1}^j,\Delta\tilde{S}_{k+1}^i\tilde{S}_{l+1}^j, \tilde{S}_{k+1}^i\tilde{S}_{l+1}^j, \tilde{S}_{k+1}^i\Delta\tilde{S}_{l+1}^j  \}$. By using the same decomposition and arguments as to prove Lemma \ref{lem:alpha} and Cauchy Schwartz' inequality, we find $\mathcal{X}\in \mathcal{L}$ such that 
\beaa
\begin{aligned}
& \E_{\Prob}\left[|\alpha_{k+1}^i(Z^{\star})-\alpha_{k+1}^i(\mathscr{C}_K(Z^{\star}))||\alpha_{l+1}^j(\mathscr{C}_K(Z^{\star}))||\mathcal{Y}|\right]=\E_{\Prob}\left[|\Delta Z^{\star}_{k+1}-\Delta \mathscr{C}_K(Z^{\star})_{k+1}||\Delta \mathscr{C}_K(Z^{\star})_{l+1}||\mathcal{X}|\right]\\
& \leq \E_{\Prob}\left[|\Delta Z^{\star}_{k+1}-\Delta \mathscr{C}_K(Z^{\star})_{k+1}|^2\right]^{\frac{1}{2}}\E_{\Prob}\left[|\Delta \mathscr{C}_K(Z^{\star})_{l+1}|^2 |\mathcal{X}|^2\right]^{\frac{1}{2}}.
\end{aligned}
\eeaa
The same result can be obtain for 
$\E_{\Prob}\left[|\alpha_{k+1}^i(Z^{\star})-\alpha_{k+1}^i(\mathscr{C}_K(Z^{\star}))||\alpha_{l+1}^j(Z^{\star})||\mathcal{Y}|\right]$.
By using the main assumption of the proposition, and considering the form of $R(\mathscr{C}_K(Z^{\star}))$ and $R(Z^{\star})$, we find $\zeta>0$ such that
$$\E_{\Prob}\left[|\mathcal{R}(\mathscr{C}_K(Z^{\star}))+\mathcal{R}(Z^{\star})||\mathcal{R}(\mathscr{C}_K(Z^{\star}))-\mathcal{R}(Z^{\star})|\right]\leq \zeta \E_{\Prob}\left[|Z^{\star}_N-\mathscr{C}_K(Z^{\star}_N)|^2\right]^{\frac{1}{2}}. 
$$
Finally, by using Lemma \ref{lem:alpha} to bound the term $2|\theta^{\star}|\E_{\Prob}\left[|\mathcal{R}(\mathscr{C}_K(Z^{\star}))-\mathcal{R}(Z^{\star})|\right]$, we find $\eta_2>0$ such that 
\beaa
\E_{\Prob}\left[|Pr(\mathscr{C}_K(Z^{\star}))_N-\mathcal{C}_N(\mathscr{C}_K(Z^{\star}))-Pr(Z^{\star})_N+\mathcal{C}_N(Z^{\star})|^2\right] \leq \eta_2\E_{\Prob}\left[|Z^{\star}_N-\mathscr{C}_K(Z^{\star}_N)|^2\right]^{\frac{1}{2}}.
\eeaa
Finally setting $\eta=\eta_1S_N^0+ \gamma\eta_2(S_N^0)^2$ gives the result. 
\end{proof}

\begin{proof}[of Proposition \ref{prop:convergence}]
Let $(Z^K,\theta^K)$ be a solution to \eqref{system_chaos_K}. $(\mathscr{C}_K(Z^K),\theta^K)$ is an admissible solution for \eqref{system_chaosT2}, then 
\be
\E_{\Prob}\left[F^{\gamma}_K(Z^{\star},\theta^{\star})\right]\leq \E_{\Prob}\left[F^{\gamma}_K(Z^K,\theta^K)\right]\leq \E_{\Prob}\left[F^{\gamma}(Z^{\star},\theta^{\star})\right]. 
\ee
According to Lemma \ref{lem:convergence}, we find $\eta>0$ independent of $K$ such that
\be
\begin{aligned}
|\E_{\Prob}\left[F^{\gamma}(Z^{\star},\theta^{\star})\right]-\E_{\Prob}\left[F^{\gamma}_K(Z^K,\theta^K)\right]|
& \leq |\E_{\Prob}\left[F^{\gamma}(Z^{\star},\theta^{\star})\right]-\E_{\Prob}\left[F^{\gamma}_K(Z^{\star},\theta^{\star})\right]| \\
& \leq \E_{\Prob}\left[|F^{\gamma}(Z^{\star},\theta^{\star})-F^{\gamma}_K(Z^{\star},\theta^{\star})|\right]\\
& \leq \eta \E_{\Prob}\left[|Z^{\star}_N-\mathscr{C}_K(Z^{\star}_N)|^2\right]^{\frac{1}{2}}. 
\end{aligned}
\ee
\end{proof}
We denote $m$ as the number of coefficients $\beta_{\lambda}$ appearing 
 in the chaos expansion of order $K$, 
$m=\#\{\lambda \in (\mathbb{N}^{N})^d |\lambda|_1\leq K\}$. By identifying the martingale $Z\in\mathcal{M}$ with the coefficients of the chaos expansion of its terminal value $Z_N$,  
we slightly abuse the notation $\mathscr{C}_K$ to write
$$
\forall \beta\in\mathbb{R}^m, \ \mathscr{C}_K(\beta)=\sum_{\substack{\lambda \in (\mathbb{N}^{N})^d \\ |\lambda|_1\leq K}} \beta_{\lambda}  H_{\lambda}^{\bigotimes}(\Delta \widehat{W}^\bQ).
$$
We can also slightly abuse the definition of the controls in (\ref{eq:control_X}). Recalling that $(\beta_\lambda)_\lambda$ define the expansion of $Z_N$, we define
\be
\label{eq:control_chaos}
\alpha_{n+1}(\beta) = \E_\bQ\left[\Delta \tilde{S}_{n+1} \Delta (\tilde{S}_{n+1})^T|\cF_{n}\right]^{-1}\E_\bQ\left[\mathscr{C}_K(\Delta {Z}_{n+1})\Delta \tilde{S}_{n+1}|\cF_{n}\right].
\ee   
We notice that $\alpha_{n+1}$ can be expressed as a linear combination of the chaos coefficients of $Z_N$. 
Similarly, we also extend the cost function to
\be
 \mathcal{C}_n(\beta)=\sum_{k=0}^{n-1}\sum_{i=1}^d \nu|\alpha^i_{k+1}(\beta)-\alpha^i_k(\beta)|\tilde{S}_k^i.
\ee
Let's extend in this context, the portfolio value function as
\be
\forall \beta\in\mathbb{R}^m, \ \mathcal{R}(\beta)=Pr\left(\mathscr{C}_K(\beta)\right)_N-\mathcal{C}_N(\beta).
\ee
Note that $\beta\rightarrow \mathcal{R}(\beta)-v_0$ is a positive homogeneous function, ie $\forall \ u>0, \ \mathcal{R}(u\beta)-v_0= u\left(\mathcal{R}(\beta)-v_0\right)$. 
We also express the objective function as 
\be
\forall \beta\in\mathbb{R}^m, \ F^{\gamma}(\beta,\theta)=\mathcal{R}(\beta)S_N^0-\gamma\left(\left(\mathcal{R}(\beta)-\theta\right)S_N^0\right)^2.
\ee
Note that $F^\gamma$ is a random function. With these new functions, we can approximate the original problem~\eqref{system_martingale2} by a finite dimensional optimisation problem
\be
\label{system_chaos}
\tag{$J^{\gamma}$}
\begin{aligned}
& \sup_{\beta\in\mathbb{R}^m, \theta\in\mathbb{R}}
& & \E_{\Prob}\left[F^{\gamma}(\beta,\theta)\right] \\
& \text{s.t.}
& & \E_\bQ\left[\mathscr{C}_{K}(\beta)\right]=v_0
\end{aligned}
\ee
The constraint $\E_\bQ\left[\mathscr{C}_{K}(\beta)\right]=v_0$ is easily satisfied by setting the first coefficient equal to $v_0$,  $\beta_{0}=v_0$.
\begin{proposition}
The problem \eqref{system_chaos} admits a solution. 
\end{proposition}
\begin{proof}
According to Lemma \ref{prop:equivalent}, for $\beta\in \mathbb{R}^m$, $\underset{\theta\in\mathbb{R} }{\sup} \ \E_{\Prob}\left[F^{\gamma}(\beta,\theta)\right]$ always exists and is attained for $\theta^{\star}=\E_{\Prob}[\mathcal{R}(\beta)]$.
Let $\beta \in  \mathbb{R}^m \neq 0_{\mathbb{R}^m}$. 
We have
\beaa
\begin{aligned}
\E_{\Prob}\left[F^{\gamma}(\beta,\E_{\Prob}[\mathcal{R}(\beta)])-v_0\right] &=\left|\beta\right|\E_{\Prob}\left[\mathcal{R}\left(\frac{\beta}{\left|\beta\right|}\right)-v_0\right]-\gamma \left|\beta\right|^2\Var\left[\mathcal{R}\left(\frac{\beta}{\left|\beta\right|}\right)-v_0\right]
\\
& = \left|\beta\right|\E_{\Prob}\left[F^{\gamma}\left(\frac{\beta}{\left|\beta\right|},\E_{\Prob}[\mathcal{R}(\frac{\beta}{\left|\beta\right|})]\right)-v_0\right]+\gamma \left(\left|\beta\right|-\left|\beta\right|^2\right)\Var\left[\mathcal{R}\left(\frac{\beta}{\left|\beta\right|}\right)\right].
\end{aligned}
\eeaa
By considering the optimums of continuous functions on compact sets, we define 
\beaa
\begin{aligned}
v=& \inf_{ \left|\beta\right|=1}
& & \ \Var\left[\mathcal{R}(\beta)\right]
\end{aligned}, \ 
\begin{aligned}
u=& \sup_{ \left|\beta\right|=1}
& & \ \E_{\Prob}\left[F^{\gamma}\left(\frac{\beta}{\left|\beta\right|},\E_{\Prob}[\mathcal{R}(\frac{\beta}{\left|\beta\right|})]\right)-v_0\right].
\end{aligned}
\eeaa
Since the market is assumed to be complete, it is not possible to build a risk free portfolio with risky assets, $v>0$.
Finally, $\forall \ \beta \in\mathbb{R}^m, \ \left|\beta\right|>1$,
\beaa
\E_{\Prob}\left[F^{\gamma}(\beta,\E_{\Prob}[\mathcal{R}(\beta)])-v_0\right] \leq \left|\beta\right| u+\gamma\left(\left|\beta\right|-\left|\beta\right|^2\right) v.
\eeaa
With this inequality, we deduce that
$
\lim \limits_{
\begin{array}{l} 
\left|\beta\right| \to +\infty
\end{array}}\E_{\Prob}\left[F^{\gamma}(\beta,\E_{\Prob}[\mathcal{R}(\beta)])\right]=-\infty.$ 
The objective function $\beta\rightarrow \E_{\Prob}\left[F^{\gamma}(\beta,\E_{\Prob}[\mathcal{R}(\beta)])\right]$ is continuous and $\E_{\Prob}\left[F^{\gamma}(0,\E_{\Prob}[\mathcal{R}(0)])\right]=v_0>0$, then it attains its maximum. 
We can conclude that \eqref{system_chaos} admits a solution. 
\end{proof}

\section{Resolution}
\label{sec:resolution}
In this section, we describe the numerical framework to solve \eqref{system_chaos}. We define the space where optimal strategies are investigated and we study the link between risk aversion and those solutions. In particular, we insure that risk aversion parametrization has no impact on performances. 
We assume for the numerical resolution that $\sigma$ is deterministic (Black-Scholes case). 

\subsection{Stochastic descent gradient algorithm}
We aim to apply a stochastic descent gradient algorithm to solve \eqref{system_chaos}. 
Let's introduce
$$ 
\mathcal{Z}=\left\{\beta\in\mathbb{R}^m, \ \forall n \in \{1,\ldots,N\}, i\in\{1,\ldots,d\} \ \alpha_n^i(\beta)\neq 0 \ a.s.\right\}
$$
$\mathbb{R}^m\setminus\mathcal{Z}$ is a closed set of $\mathbb{R}^m$ of measure null. Elements of this set represents strategies where there is almost surely, no investment, during a time period between $t_1$ and $t_{N-1}$, in a risky asset. These strategies are not realistic, not optimal and therefore are ignored. Optimal strategies are investigated on $\mathcal{Z}$.
It can be shown that $(\beta,\theta)\longmapsto\E_{\Prob}[F^{\gamma}(\beta,\theta)]=\E_{\Prob}\left[\mathcal{R}(\beta)S_N^0-\gamma\left(\left(\mathcal{R}(\beta)-\theta\right)S_N^0\right)^2\right]$ is differentiable almost everywhere, on $\mathcal{Z}\times \mathbb{R}$. We refer to Appendix \ref{sec:differentiability} for a detailed study on differentiability.
Furthermore, we have seen that
$
\E_{\Prob}\left[\nabla_{\theta}F^{\gamma}(\beta,\theta^{\star})\right]=0 \iff \theta^{\star}=\E_{\Prob}\left[\mathcal{R}(\beta)\right].
$
We can then deduce that if $\beta^{\star}\in\mathcal{Z}$ is the chaos expansion decomposition of an optimal portfolio then $\beta^{\star}$ is solution to 
\be
\label{eq:grad}
\tag{$T^{\gamma}$}
\E_{\Prob}\left[\nabla_{\beta,\theta}F^{\gamma}\left(\beta^{\star},\E_{\Prob}\left[\mathcal{R}(\beta^{\star})\right]\right)\right]=0.
\ee
\begin{proposition}
The solutions to \eqref{eq:grad} are locally optimal.
\end{proposition}
\begin{proof}
$\mathcal{R}$ is concave w.r.t. $\beta$. $F_{y}:(Y,\theta) \in \bR^2 \rightarrow\E_{\Prob}\left[YS_N^0-\gamma\left(\left(Y-\theta\right)S_N^0\right)^2\right]$ is concave. For $\theta\in\mathbb{R}$, if $\E_{\Prob}[F_{y}(.,\theta)]$ is monotone increasing then $\E_{\Prob}[F^{\gamma}(\beta,\theta)]$ is concave in $\beta$. 
$\E_{\Prob}[F_{y}(.,\theta)]$ is increasing in $Y$ if and only if $\E_{\Prob}[Y]\leq \frac{1}{2\gamma}+\theta$. 
\\
Let $\beta^{\star}$ be a solution to \eqref{eq:grad}. 
$\mathcal{R}$ is continuous then we can find $\varepsilon>0$ such that $\forall \ \beta\in\mathbb{R}^m$ s.t. $|\beta-\beta^{\star}|\leq \varepsilon$, we have $|\mathcal{R}(\beta)-\mathcal{R}(\beta^{\star})|\leq \frac{1}{4\gamma}$. 
Then $\forall  (\beta,\theta)$ such that 
$|\beta-\beta^{\star}|\leq \varepsilon$, $|\theta-\E_{\Prob}[\mathcal{R}(\beta^{\star})]| \leq \frac{1}{4\gamma}$, we have 
$$
|\theta-\E_{\Prob}[\mathcal{R}(\beta)]|\leq |\theta-\E_{\Prob}[\mathcal{R}(\beta^{\star})]| + \E_{\Prob}[|\mathcal{R}(\beta^{\star})-\mathcal{R}(\beta)|]\leq \frac{1}{2\gamma}. 
$$
Then $\E_{\Prob}[F^{\gamma}(\beta,\theta)]$ is concave on $\{ (\beta, \theta) \, : \, |\beta - \beta^\star| \le \varepsilon, \; |\theta - \E_{\Prob}[\mathcal{R}(\beta^{\star})]| \le \frac{1}{4 \gamma} \}$.
\end{proof}
Finally we can apply a gradient descent algorithm to find a local optimum of \eqref{system_chaos}.

\subsection{Influence of risk aversion}
\label{sec:aversion}
In this section, we discuss the influence of the risk aversion on optimal multi-period portfolios submitted to transaction costs. 
 One main result claims that the Sharp ratio of a zeros gradient portfolio submitted to transaction costs is independent from its risk aversion. A second main result affirms that all optimal portfolios have the same Sharp ratio. This study insures that the choice of risk aversion parameter does not influence performances. 
 \\
In this section, we consider the following formulation of the allocation problem
\be
\label{system_chaos2}
\tag{$\mathcal{J}^{\gamma}$}
\begin{aligned}
& \sup_{\beta\in\mathcal{Z}}
& & \E_{\Prob}\left[G^{\gamma}(\beta)\right] \\
& \text{s.t.}
& & \E_\bQ\left[\mathscr{C}_{K}(\beta)\right]=v_0
\end{aligned}
\ee
where $G^{\gamma}(\beta)=F^{\gamma}(\beta,\E_{\Prob} [\mathcal{R}(\beta)])$. 
We have seen that $(\beta^{\star},\theta^{\star})$ solves \eqref{system_chaos} if and only if $\theta^{\star}=\E_{\Prob}[\mathcal{R}(\beta^{\star})]$ and $\beta^{\star}$ solves \eqref{system_chaos2}. 

\begin{definition}
A $\gamma$-optimal strategy is a strategy $(\alpha_n(\beta^{\star}))_{0 \le n \le N}$ such that $\beta^{\star}$ solves \eqref{system_chaos2}. The portfolio associated to a $\gamma$-optimal strategy
is called a $\gamma$-optimal portfolio.
\end{definition}
\begin{definition}
A $\gamma$-zero gradient strategy or $\gamma$-locally optimal strategy is a strategy $(\alpha_n(\beta^{\star}))_{0 \le n \le N}$ such that $\beta^{\star}$ solves \eqref{eq:grad}. The portfolio associated to a $\gamma$-zero gradient strategy
is called a $\gamma$-zero gradient portfolio. 
\end{definition}  
\begin{proposition}
\label{prop:sharp-gradient}
 The risk aversion, the Sharpe ratio and the volatility of a $\gamma$-zero gradient portfolio $V^{\star}$ submitted to transaction costs are related by 
\be
\label{eq:sharp}
\gamma=\frac{\sharpe(V_N^{\star})}{2\Var[V_N^{\star}]^{\frac{1}{2}}}. 
\ee
\end{proposition}
\begin{lemma}
\label{prop:R}
 There exist a $\cF_N$-measurable random variable $D$ and $\bfF$-adapted processes $B^i,K^i$ for $i\in\{1,\dots,d\}$ and taking values in $\mathbb{R}^m$, such that $\forall i\in\{1,\dots,d\}, \ n\in\{1,\dots,N\}$, $\E_{\Prob}\left[|B_n^i|^p\right] + \E_{\Prob}\left[|K_n^i|^p\right] + \E_{\Prob}\left[|D|^p\right] <\infty$ for any $p \ge 1$ and
 \be
\alpha_{n}^i(\beta)=B_n^i\cdot\beta; \ \mathcal{R}(\beta)=v_0+D\cdot\beta-\sum_{i=1}^d\sum_{n=0}^{N-1}\nu |K_n^i\cdot \beta|.
 \ee
\end{lemma}
\begin{proof}[of Lemma \ref{prop:R}]
According to the linearity in $\beta$ of the control functions, defined in (\ref{eq:control_chaos}), there exists a $\cF_N$-measurable random variable $D$ with values in $\mathbb{R}^m$ such that
$
\Pr\left(\mathscr{C}_K(\beta)\right)_N=v_0+D\cdot\beta$
and  $\exists \ \{B^i, i\in\{1,\dots, d\}\}$,  $\bfF$-adapted processes with values in $\mathbb{R}^m$ such that $\forall i\in\{1,\dots,d\}, \ n\in\{0,\dots,N-1\}$, $\alpha_{n}^i(\beta)=B_n^i\cdot\beta.$
By calling $K_n^i=B_{n+1}^i-B_n^i$, we obtain
$\mathcal{C}_N(\beta^{\star})=\sum_{i=1}^d\sum_{n=0}^{N-1}\nu |K_n^i\cdot \beta|.$ 
\\
$\forall i\in\{1,\dots,d\}, \ n\in\{0,\dots,N-1\}$, $j\in\{1,\dots,m\}$, $(B_n^i)_j\in \spane\left\{H_{\lambda}^{\bigotimes}(\Delta \widehat{W}^\bQ), \lambda\in \mathbb{R}^m\right\}$,  $(K_n^i)_j\in \spane\left\{H_{\lambda}^{\bigotimes}(\Delta \widehat{W}^\bQ), \lambda\in \mathbb{R}^m\right\}$ so $\forall \ p\geq 1$,
$\E_{\Prob}\left[|B_n^i|^p\right]$ and $\E_{\Prob}\left[|K_n^i|^p\right]<\infty$. And we have $\forall \ j\in\{1,\dots,m\}$, $D_j=\sum_{n=0}^{N-1}\sum_{i=1}^d (B_{n+1}^i)_j\Delta \tilde{S}_{n+1}^i$. $\tilde{S}_n\in{L}^p(\Omega,\cF_n,\Prob)$, then $\E_{\Prob}[|D_j|^p]<\infty$. 
\end{proof}
\begin{proof}[of Proposition \ref{prop:sharp-gradient}]
We have seen in (\ref{eq:grad}) that if $V^{\star}$ is a $\gamma$-zero gradient portfolio, then $\exists$ $\beta^{\star}\in \mathcal{Z}$ such that $V_N^{\star}=\mathcal{R}(\beta^{\star})S_N^0$ and $
\E_{\Prob}\left[\Psi\left(\beta^{\star},\E_{\Prob}\left[\mathcal{R}(\beta^{\star})\right]\right)\right]=0.
$ 
This equality leads to 
$$
\E_{\Prob}\left[\nabla \mathcal{R}\left(\beta^{\star}\right)S_N^0\left(1-2 \gamma S_N^0\left(\mathcal{R}\left(\beta^{\star}\right)-\E_{\Prob}\left[\mathcal{R}\left(\beta^{\star}\right)\right]\right)\right)\right]=0.
$$
We recall that 
$\mathcal{R}(\beta^{\star})=\Pr\left(\mathscr{C}_K(\beta^{\star})\right)-\mathcal{C}_N(\beta^{\star})$.
As  $\mathcal{R}$ is a.s differentiable and using the decomposition of Lemma \ref{prop:R},
\be
\forall \ \beta\in\mathcal{Z}, \ \nabla \mathcal{R}(\beta)=D-\sum_{i=1}^d\sum_{n=0}^{N-1}\nu \times \sign(K_n^i\cdot\beta)K_n^i.
\ee 
We deduce that 
$$
\forall \ \beta\in\mathcal{Z}, \ \beta\cdot\nabla \mathcal{R}(\beta)=\mathcal{R}(\beta)-v_0.
$$
Then with
$$ \beta^{\star}\cdot\E_{\Prob}\left[\nabla \mathcal{R}\left(\beta^{\star}\right)S_N^0\left[1-2 \gamma S_N^0\left(\mathcal{R}\left(\beta^{\star}\right)-\E_{\Prob}\left[\mathcal{R}\left(\beta^{\star}\right)\right]\right)\right]\right]=0,
$$
we obtain 
$$
\E_{\Prob}\left[ \mathcal{R}\left(\beta^{\star}\right)S_N^0\left[1-2 \gamma S_N^0\left(\mathcal{R}\left(\beta^{\star}\right)-\E_{\Prob}\left[\mathcal{R}\left(\beta^{\star}\right)\right]\right)\right]\right]=v_0 S_N^0.
$$
We deduce that 
$$
\E_{\Prob}\left[\mathcal{R}\left(\beta^{\star} \right)S_N^0\right]-2\gamma \E_{\Prob}\left[\left(\mathcal{R}\left(\beta^{\star}\right)S_N^0\right)^2\right]+2\gamma \E_{\Prob}\left[\mathcal{R}\left(\beta^{\star} \right)S_N^0\right]^2=v_0 S_N^0.
$$
Finally,
\be
\gamma=\frac{\sharpe(V_N^{\star})}{2\Var[V_N^{\star}]^{\frac{1}{2}}}.
\ee
\end{proof}
\begin{corollary}
\label{prop:global-sharp}
All $\gamma$-optimal portfolios have the same Sharpe ratio. 
\end{corollary}
\begin{proof}
    Let $V_N^{\star}$ be such a $\gamma$-optimal portfolio. As an optimal portfolio, $V_N^{\star}$ is also a $\gamma$-zero gradient portfolio. Then according to Proposition \ref{prop:sharp-gradient},
    $$
    \E_{\Prob}\left[V_N^{\star}\right]-v_0S_N^0=2\gamma \Var\left[V_N^{\star}\right]
    $$ 
    Inserting this term in the mean-variance objective function leads to
    \be
    \label{eq:same-expectation}
    \E_{\Prob}\left[{V}_{N}^{\star}-\gamma\left({V}_{N}^{\star}-\E_{\Prob}[{V}_{N}^{\star}]\right)^2\right]=\frac{\E_{\Prob}\left[V_N^{\star}\right]}{2}-\frac{v_0S_N^0}{2}.
    \ee    
We deduce that if $V_N'$ is another $\gamma$-optimal portfolio then (\ref{eq:same-expectation}) gives $\E_{\Prob}[V_N^{\star}]=\E_{\Prob}[V_N']$. 
But the two portfolios also attain the same mean-variance value then
$$
\E_{\Prob}\left[{V}_{N}^{\star}\right]-\gamma \Var\left[{V}_{N}^{\star}\right]=\E_{\Prob}\left[{V}_{N}'\right]-\gamma \Var\left[{V}_{N}'\right].
$$
With this last equality, we can conclude they have the same variance and as a result, the same Sharpe ratio. 
\end{proof}
\begin{remark}
Maximizing the mean-variance objective function of zero gradient portfolios, is equivalent to maximizing their Sharpe ratio.
\end{remark}
Now, let's prove this important proposition on the impact of the risk aversion.
\begin{proposition}
\label{prop:beta_gamma-zero}
Let $(\gamma,\gamma')$ $\in$ $\mathbb{R}^{*}$, then $(\alpha_n^{\star})_n$ is a $\gamma$-zero gradient strategy if and only if ($\frac{\gamma}{\gamma'}\alpha_n^{\star})_n$ is a $\gamma'$-zero gradient strategy.
\end{proposition}
\begin{proof}
Let $(\alpha_n^{\star})_n$ be a $\gamma$-zero gradient strategy with the associated chaos coefficients $\beta^{\star}$.
Then $\beta^{\star}$ is solution to \eqref{eq:grad}. Then with the previous notations
\be
\label{eq:grad_utility}
\E_{\Prob}\left[\nabla \mathcal{R}\left(\beta^{\star}\right)S_N^0\left[1-2 \gamma S_N^0\left(\mathcal{R}\left(\beta^{\star}\right)-\E_{\Prob}\left[\mathcal{R}\left(\beta^{\star} \right)\right]\right)\right]\right]=0.
\ee
By using the decomposition of Lemma \ref{prop:R}, 
$\mathcal{R}$ is a.s differentiable on $\mathcal{Z}$ and 
\be
\label{eq:grad_portfolio}
\nabla \mathcal{R}(\beta)=D-\sum_{i=1}^d\sum_{n=0}^{N-1}\nu \times \sign(K_n^i\cdot\beta)K_n^i.
\ee
Let's now consider the portfolio 
$$
\mathcal{R}(\frac{\gamma}{\gamma'}\beta^{\star})=v_0+\frac{\gamma}{\gamma'}D\cdot\beta^{\star}-\sum_{i=1}^{d}\sum_{n=0}^{N-1}\nu|K_n^i\cdot\frac{\gamma}{\gamma'}\beta^{\star}|.
$$
We obtain the following equality
$$
\mathcal{R}(\frac{\gamma}{\gamma'}\beta^{\star})-\frac{\gamma}{\gamma'}\mathcal{R}(\beta^{\star})=v_0(1-\frac{\gamma}{\gamma'}).
$$
We deduce that
\be
\label{eq:exp}
\frac{\gamma}{\gamma'}\left(\mathcal{R}(\beta^{\star})-\E_{\Prob}\left[\mathcal{R}(\beta^{\star})\right]\right)=\mathcal{R}(\frac{\gamma}{\gamma'}\beta^{\star})-\E_{\Prob}\left[\mathcal{R}(\frac{\gamma}{\gamma'}\beta^{\star})\right].
\ee
By using the expression (\ref{eq:grad_portfolio}), we also notice that 
\be
\label{eq:grad_gamma}
\nabla \mathcal{R}(\frac{\gamma}{\gamma'}\beta^{\star})=D-\sum_{i=1}^d\sum_{n=0}^{N-1}\nu \times \sign\left(K_n^i\cdot\frac{\gamma}{\gamma'}\beta^{\star}\right)K_n^i=\nabla \mathcal{R}(\beta^{\star}).
\ee
Finally, by inserting (\ref{eq:exp}) and (\ref{eq:grad_gamma}) in equality (\ref{eq:grad_utility}), we obtain
\be
\E_{\Prob}\left[\nabla \mathcal{R}\left(\frac{\gamma}{\gamma'}\beta^{\star}\right)S_N^0\left[1-2 \gamma'S_N^0\left(\mathcal{R}\left(\frac{\gamma}{\gamma'}\beta^{\star}\right)-\E_{\Prob}\left[\mathcal{R}\left(\frac{\gamma}{\gamma'}\beta^{\star} \right)\right]\right)\right]\right]=0,
\ee which finishes the proof.
\end{proof}
\begin{corollary}
\label{prop:beta_gamma}
Let $(\gamma,\gamma')$ $\in$ $\mathbb{R}^{*}$, then $(\alpha_n^{\star})_n$ is a $\gamma$-optimal strategy if and only if ($\frac{\gamma}{\gamma'}\alpha_n^{\star})_n$ is a $\gamma'$-optimal strategy.
\end{corollary}
\begin{proof}
Let $(\alpha_n^{\star})_n$ be a $\gamma$-optimal strategy with the associated chaos coefficients $\beta^{\star}$.  
Then $\beta^{\star}$ is a solution to \eqref{system_chaos2}. 
Firstly, according to Proposition \ref{prop:beta_gamma-zero}, $\frac{\gamma}{\gamma'}\beta^{\star}$ is a solution to $(T^{\gamma'})$ with 
\beaa
\E_{\Prob}[G^{\gamma'}(\frac{\gamma}{\gamma'}\beta^{\star})]=\E_{\Prob}\left[\mathcal{R}(\frac{\gamma}{\gamma'}\beta^{\star})S_N^0-\gamma'\left(\left(\mathcal{R}(\frac{\gamma}{\gamma'}\beta^{\star})-\E_{\Prob}\left[\mathcal{R}(\frac{\gamma}{\gamma'}\beta^{\star})\right]\right)S_N^0\right)^2\right].
\eeaa
But as $\beta\rightarrow \mathcal{R}(\beta)-v_0$ is a positive homogeneous function, ie $\forall \ u>0, \ \mathcal{R}(u\beta)-v_0= u\left(\mathcal{R}(\beta)-v_0\right)$, then 
\beaa
\begin{aligned}
\E_{\Prob}[G^{\gamma'}(\frac{\gamma}{\gamma'}\beta^{\star})]-v_0S_N^0
& =\E_{\Prob}\left[\frac{\gamma}{\gamma'}\left(\mathcal{R}(\beta^{\star})-v_0\right)S_N^0-\gamma'\left(\frac{\gamma}{\gamma'}\left(\mathcal{R}(\beta^{\star})-\E_{\Prob}\left[\mathcal{R}(\beta^{\star})\right]\right)S_N^0\right)^2\right]\\
& =\frac{\gamma}{\gamma'}\left(\E_{\Prob}[G^{\gamma}(\beta^{\star})]-v_0S_N^0\right).
\end{aligned}
\eeaa
If $\exists \ \beta^{1}\in\mathbb{R}^m$ such that 
$\E_{\Prob}[G_{\gamma'}(\frac{\gamma}{\gamma'}\beta^{\star})]<\E_{\Prob}[G_{\gamma'}(\beta^{1})]
$ then using the previous equality $\E_{\Prob}[G^{\gamma}(\beta^{\star})]<\E_{\Prob}[G_{\gamma'}(\frac{\gamma'}{\gamma}\beta^{1})]$, which is impossible with the optimally of $\beta^{\star}$. We deduce that
$\frac{\gamma}{\gamma'}\beta^{\star}$ is a solution to $(\mathcal{J}^{\gamma'})$. The converse statement is obvious. 
\\
We have seen in (\ref{eq:control_chaos}), that the controls 
$(\alpha_n^{\star})_{n}$ can be expressed as  linear functions of $\beta^{\star}$. We can conclude that $(\alpha_n^{\star})_n$ is an optimal solution for \eqref{sys:cost2} if and only if ($\frac{\gamma}{\gamma'}\alpha_n^{\star})_n$ is an optimal solution for $(E^{\gamma'})$.  
\end{proof}
\begin{proposition}
\label{prop:same-sharp}
Let $(\alpha_n^{\star})_n$ be a $\gamma$-zero gradient strategy and $(V_n^{\star})_n$ the associated $\gamma$-zero gradient portfolio. Then for all $u>0$,
$\sharpe(V_N^{\star})=\sharpe(V_N^{u})$,
where $(V_n^u)_n$ is the portfolio associated to the strategy $(u\alpha_n^{\star})_n$.
\end{proposition}
\begin{proof}
By using that $\beta\rightarrow \mathcal{R}(\beta)-v_0$ is a positive homogeneous function, we have
\beaa
\begin{aligned}
\sharpe(V_N^{\star}) &=\frac{\E_{\Prob}\left[\tilde{V}_{N}^{\star}S_N^0\right]-v_0S_N^0}{\Var\left[\tilde{V}^{\star}_{N}S_N^0\right]^{\frac{1}{2}}}
=\frac{u\E_{\Prob}\left[\mathcal{R}(\beta^{\star})S_N^0\right]}{u \Var\left[\mathcal{R}(\beta^{\star})S_N^0\right]^{\frac{1}{2}}}
= \sharpe(V_N^{u}), 
\end{aligned}
\eeaa
that finishes the proof. 
\end{proof}
\begin{proposition}
\label{lemma2}
The Sharpe ratio of a zeros gradient portfolio is independent of its risk aversion.
\end{proposition}
\begin{proof}
A direct consequence of Propositions \ref{prop:beta_gamma-zero} and \ref{prop:same-sharp}. 
\end{proof}
\begin{proposition}
\label{prop:benchmark}
Let $\mathcal{V}>0$ and  $(V_n^{\star})_n$ be a $\gamma$-zero gradient portfolio. Then there exists $\gamma'$ and a $\gamma'$-zeros gradient portfolio $V'$ with $\sharpe(V_N')=\sharpe(V_N^{\star})$ and with volatility $\mathcal{V}$.
\begin{proof}
 According to Propositions \ref{prop:beta_gamma-zero} and \ref{prop:same-sharp}, for all $u>0$, $(u\alpha_{n}^{\star})_n$ is a $\frac{\gamma}{u}$-zeros gradient strategy
and the associated $\frac{\gamma}{u}$-zeros gradient portfolio $V_N^{u}$ has the same Sharpe ratio as $V_N^{\star}$. Consequently, according to Proposition \ref{prop:sharp-gradient}, $$
 \frac{\gamma}{u}=\frac{\sharpe(V_N^{\star})}{2\Var(V_N^u)^{\frac{1}{2}}}.
 $$
 Then by choosing $u=\frac{2\gamma\mathcal{V}}{\sharpe(V_N^{\star})}$, the portfolio $(u\alpha_{n}^{\star})_n$ verifies the conditions. 
\end{proof}
\end{proposition}
\begin{proposition}
\label{prop:multi-sharp-cost}
Let $\gamma, \ \gamma'>0$ s.t. $\gamma\neq \gamma'$. Let $V^{\star}$ (resp. $V'$) be a $\gamma$-optimal portfolio (resp. $\gamma'$-optimal portfolio). Then, $\sharpe(V_N^{\star})=\sharpe(V_N')$.
\end{proposition}
\begin{proof}
Let $(\alpha_n')_n$ be the $\gamma'$-optimal strategy associated to $V'$. Let $V^{\#}$ the portfolio associated to the strategy $(\frac{\gamma'} {\gamma}\alpha_n')_n$. According to Proposition \ref{prop:beta_gamma}, $V^{\#}$ is a $\gamma$-optimal portfolio. Furthermore, using Proposition \ref{prop:same-sharp}, we have
$\sharpe(V_N')=\sharpe(V_N^{\#})$. 
According to Proposition \ref{prop:global-sharp}, $\gamma$-optimal portfolios have the same Sharpe ratio. So $\sharpe(V_N^{\star})=\sharpe(V_N^{\#})$, and finally $\sharpe(V_N^{\star})=\sharpe(V_N')$. 
\end{proof}

\section{Numerical illustration}
\label{sec:result}
In this section, we implement the approach presented above and the sequential uni-period Markowitz model which is used as benchmark, see~\ref{sec:uni}. We consider the cases where costs are ignored and considered. 
In order to maximise the rate of return of a portfolio while controlling its volatility at time $T$, an agent consecutively maximizes at time $0=t_0<\cdots<t_N=T$, its uni-period mean-var objective function. We assume that the agent has a uni-period risk aversion parameter $\gamma_u$. The agent has to consecutively solve for $n \in \{0,\dots,N-1\}$,
\be
\label{system_uniperiod}
\tag{$Y^{\gamma_u}_{n}$}
\begin{aligned}
\ & \sup_{\alpha_{n+1}, \alpha_{n+1}^0\in\mathbb{R}^d\times\mathbb{R}}
& & \E_{\Prob}\left[{V}_{n+1}|\cF_{n}\right]-\gamma_u \Var_{\Prob}\left[V_{n+1}|\cF_{n}\right] \\
& \text{subject to}
& & \Delta V_{n+1}=\alpha_{n+1}\cdot \Delta S_{n+1}+\alpha_{n+1}^0 \Delta S_{n+1}^0 -\sum_{i=1}^d\nu|\alpha_{n+1}^i-\alpha_{n}^i|S_{n}^i
\end{aligned}
\ee
Further details on the implemented benchmark model and its link with risk aversion are presented in Appendix \ref{sec:uni}.
\\
We assume that the risky assets follow the dynamics (\ref{eq:asset}) with constant parameters. Formally, the drift terms, the volatility matrices and the risk free rate are chosen deterministic and constant, ie. $\forall \ n\in\{0,\dots,N\} \ \mu_n=\mu,\ \sigma_n=\sigma, \ r_n=r$. 
Along a first part, we compare the performances by choosing the same risk aversion parameter $\gamma=\gamma_u$. As the Sharpe ratios of the estimated portfolios is independent of the aversion parameter (see Propositions \ref{lemma2}, \ref{prop:multi-sharp-cost} , \ref{prop:uni-sharp}, \ref{prop:uni-sharp-cost}), it can be used as an indicator of performance. Then, along a second part, we apply a framework for matching the risk aversion between uni-period and multi-period models. We are able to illustrate and compare the behaviour of our solutions on two realisations. 

\subsection{Model Parameters}
We consider $d=3$ assets, evolving during $N=368$ days. Transactions are only available every $92$ days. The model is described in Section \ref{sec:environment}. Instead of specifying a volatility matrix, we fix the marginal volatilities $(\hat{\sigma}^i)_{i\in\{1.2,3\}}$ and a correlation matrix $\rho$ as in Tables~\ref{tab:parameters} and~\ref{tab:correlation}. 
\begin{table}[H]
    \centering
    \caption{Model parameters}
    \begin{tabular}{|*{5}{c|}}
        \hhline{~*{3}{-}}
        \multicolumn{1}{c|}{} & $S_1$ & $S_2$ & $S_3$     \\ \hline
        $\mu$ & {0.06} & {0.02} & {0.14}    \\ \hline
        $\hat{\sigma}$ & {0.1} & {0.06} & {0.2}   \\ \hline
    \end{tabular}
     \label{tab:parameters}
\end{table}
\begin{table}[H]
    \centering
    \caption{Correlation matrix $\rho$ of risky assets}
    \begin{tabular}{|*{5}{c|}}
        \hhline{~*{4}{-}}
        \multicolumn{1}{c|}{} & $S_1$ & $S_2$ & $S_3$     \\ \hline
        $S_1$ & {1} & {-0.2} & {0.3}    \\ \hline
        $S_2$ & {-0.2} & {1} & {-0.2}   \\ \hline
       $S_3$ & {0.3} & {-0.2} & {1}   \\ \hline
    \end{tabular}
     \label{tab:correlation}
\end{table}
We fix a constant risk free rate $r=0.001$. The initial portfolio wealth is $v_0=100$. The implementation parameters are summarized in the following table
\begin{table}[H]
    \centering
    \caption{Implementation parameters}
    \begin{tabular}{|*{6}{c|}}
        \hhline{*{6}{-}}
          nb.traj. calibration &
         nb.traj. test & N &  p & Chaos degree    \\ \hline
          {$10^5$} & {$10^5$} &368& 92 & 2  \\ \hline
    \end{tabular}
     \label{tab:implementation_parameters_multi}
\end{table}
We split our sample into two parts. The first part is used to run the descent gradient algorithm to calibrate and to find an optimal portfolio, while the second part is used to compute performance indicators.
In presence of costs, we assume that the first position is free of charge.

\subsection{Same risk aversion parameter}
\label{sec:result-with-cost}
In this first experiment, we choose the same risk aversion parameter $\gamma=\gamma_u$ for the different approaches. Since the objective functions are not the same between multi-period and uni-period models, the agents have not the same risk aversion. Therefore volatility and rates of return are not comparable. Nevertheless the Sharpe ratio is a relevant indicator for comparing performances of optimal portfolios based on different risk aversions. Indeed, the Sharpe ratios of estimated portfolios are independent from the risk aversion in every models according to Propositions \ref{lemma2}, \ref{prop:multi-sharp-cost}, \ref{prop:uni-sharp} and \ref{prop:uni-sharp-cost}. 
 \\
The implementation parameters are summarized in the following table.
\begin{table}[H]
    \centering
    \caption{Implementation parameters for Multiperiod model}
    \begin{tabular}{|*{6}{c|}}
        \hhline{*{6}{-}}
          risk aversion & batch size & iteration  & learning rate & cost(\%)  \\ \hline
          0.05 & {100}  & 1000 & 8.5 & 1 \\ \hline
    \end{tabular}
     \label{tab:implementation_parameters_multi_cost}
\end{table}
We compare five models; two sequential uni-period versions and two multi-period versions, where cost are on one hand ignored and on the other, considered. Moreover, we add in the benchmark the famous equal weight portfolio to measure the performance of the other approaches. We refer to our approach as \textit{Multi-period with costs}.  
We refer to Appendix~\ref{sec:benchmark} for further details on the four other benchmark models. 
Formally, our aim is to evaluate the performance of our method compared to the more basic approaches commonly used.
The main results can be found in Table \ref{tab:res_square_form}.
\begin{table}[H]
    \centering
    \caption{Estimated metrics for evaluating models performances}
    \begin{tabular}{|*{6}{c|}}
        \hhline{~*{5}{-}}
        \multicolumn{1}{c|}{} & rate of return($\%$) & vol($\%$) & Min-Var & Sharpe ratio      \\ \hline
        Multi-period ignoring cost & {13.24 } & {12.17 } & {105.83087 } &  {1.07939 }    \\ \hline
 Multi-period with costs & {12.31 } & {11.00 } & {106.26365 } &  {1.11033 }    \\ \hline

 Sequential uni-period ignoring cost& {10.44 } & {10.00 } & {105.43191 } &  {1.03316 }    \\ \hline
 Sequential uni-period with costs & {11.00 } & {10.72 } & {105.24599 } &  {1.01606 }    \\ \hline
equal weight & {5.69 } & {5.70 } & {104.06893 } &  {0.98125 }      \\ \hline
    \end{tabular}
     \label{tab:res_square_form}
\end{table}
We analyse here, the difference in Sharpe ratios. 
By focusing on uni-period models, it is interesting to notice that ignoring costs seems to be better than considering them. The myopic effect, joined with the consideration of costs, may make the optimal strategy rigid and inflexible. The agent's myopic behavior explains why they do not see the benefit of paying costs for short-term positions.
Therefore considering cost almost freezes the strategy and can explain that performances are not as good as if we have ignored them. 
Obviously these remarks, are dependent of the chosen parameter $\nu$, which determines the weight of costs in transactions. 
\\
In contrast, considering costs in the multi-period version is a significant improvement.  According to the choices of parameters, a difference of $0.03$ in Sharpe ratio is not negligible.
A multi-period model targets a final value and must adapt its positions according to the variations of the environment. These changes in positions are typically more significant than in myopic strategies, and thus, the impact of costs becomes more critical. Ignoring costs can have a considerable impact on strategies, leading to performance deterioration.    \\
Undoubtedly, the myopic effect has a negative impact on performances. The difference in Sharpe ratios between uni-period and multi-period models is not negligible. Therefore, we advise to use multi-period models, despite their greater complexity. In that case, costs must not be ignored.

\subsection{Same risk aversion level}
In a second experiment, we want to compare the optimal portfolios submitted to transaction costs in uni-period and multi-period settings with the same level of risk. According to Proposition \ref{prop:benchmark}, we can link the risk aversions $\gamma$ and $\gamma_u$ in both models to ensure the same level of risk.
This experiment aims to illustrate the comparison in terms of rates of returns. We can also directly compare portfolio trajectories. 

\subsubsection{Performances}
Our objective is to obtain a locally optimal multi-period portfolio that carries the same level of risk as the optimal sequential uni-period portfolio. We adopt the approach described in Proposition \ref{prop:benchmark}.  We use the $\gamma=0.05$ locally optimal multi-period portfolio estimated in the previous section to build another locally optimal multi-period portfolio with the same volatility as the optimal uni-period portfolio. 
With a uni-period risk aversion $\gamma_u=0.05$, we have obtained a volatility of $10.72 \%$ for the optimal uni-period portfolio. The $\gamma_u$ locally optimal multi-period portfolio estimated in the previous experiment has a Sharpe ratio equal to $1.11033$. According to Proposition \ref{prop:benchmark}, we can build another locally optimal multi-period portfolio with the same Sharpe ratio and a volatility of $10.72 \%$. 
To reach this volatility, 
a multi-period risk aversion parameter $\gamma=0.0518$ is estimated by using~(\ref{eq:sharp}). 
The main results can be found in Table \ref{tab:result-comp1}
\begin{table}[H]
    \centering
    \caption{Results}
    \begin{tabular}{|*{6}{c|}}
        \hhline{~*{5}{-}}
        \multicolumn{1}{c|}{} & risk aversion & rate of return($\%$) & Min-Var & Sharpe ratio      \\ \hline
       
Multi-period with costs & {0.0518 } & {11.89 } & {105.9372 } &  {1.11033 }    \\ \hline
 Sequential uni-period with costs & {0.0500 } & {11.00 } & {105.24599 } &  {1.01606 }    \\ \hline
    \end{tabular}
     \label{tab:result-comp1}
\end{table}
The multi-period model is obviously the most performing. The difference between the rates of return is almost 1\%. This difference is important according to a level of risk of 12.72\%. 
\\
After comparing the performances of the optimal portfolios with the same risk aversion, we illustrate their behaviour by showing trajectories from two different situations, A and B. To ensure consistency, we use the same framework as before to match the risk aversions.
\subsubsection{Behaviour on realisation A}
To illustrate the comparison, we present in Figures \ref{fig:asset-A}, one particular sample path of assets.
\begin{figure}[H]
\begin{center}
\includegraphics[scale=0.75]{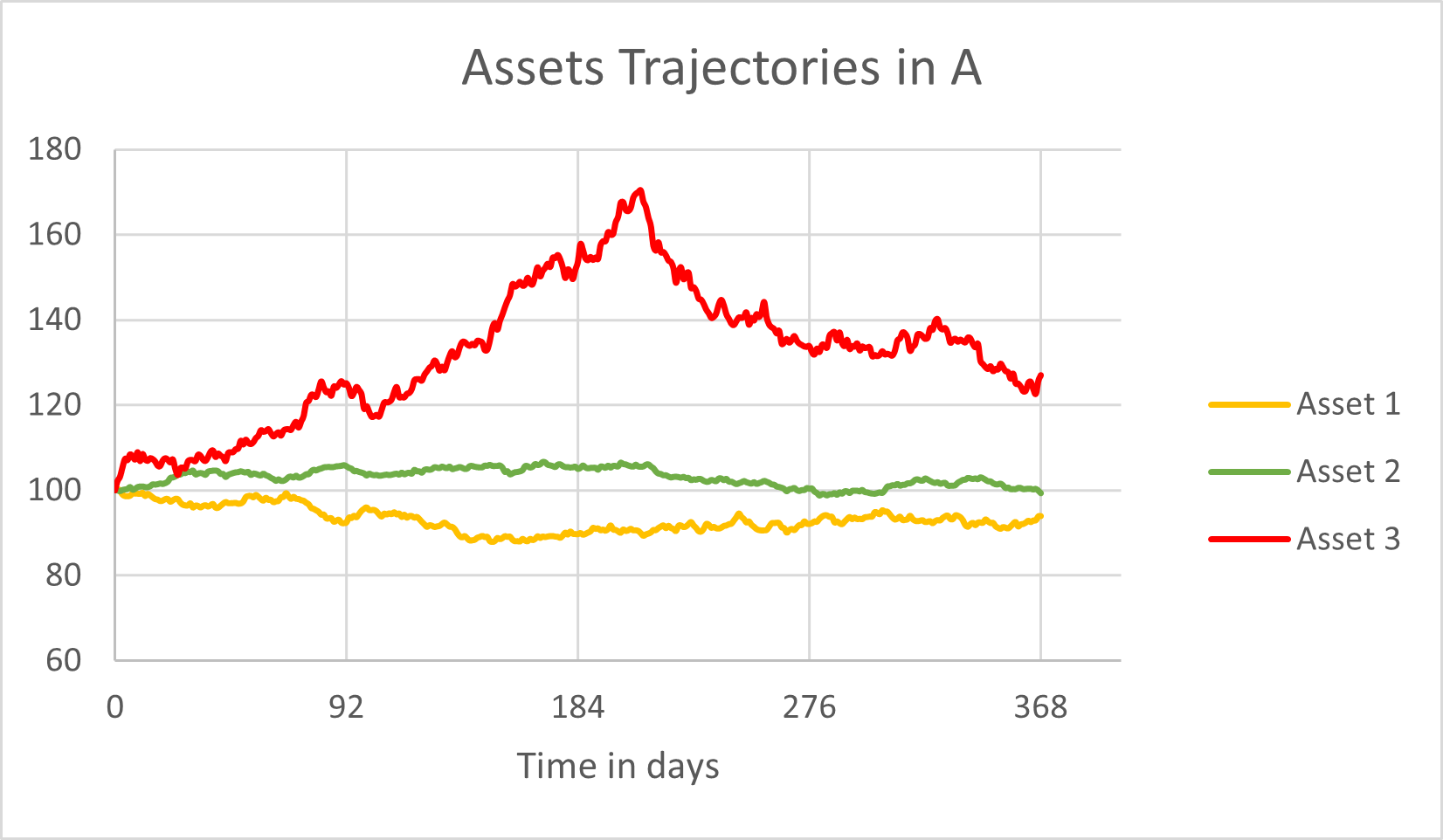}
      \caption{Asset trajectories in A}
      \label{fig:asset-A}
      \end{center}
\end{figure}
It can be observed that the curves for assets 1 and 2 are relatively flat and symmetric. Asset 2 remains above its initial value. On the other hand, asset 1 remains below its initial value. Asset 3 experiences significant growth, reaching a peak of $170$ before collapsing towards the end of the period to finish just above $120$.

We respectively present in Figures \ref{fig:portfolio-A-comp1}, \ref{fig:cost-A-comp1}, \ref{fig:control-A-uni-comp1}, \ref{fig:control-A-multi-comp1}, the portfolios values, the cumulative costs and the controls of the uni-period model considering cost against the multi-period model considering cost, whose performances have been presented in Table \ref{tab:result-comp1}.

\begin{figure}[H]
\begin{minipage}[c]{0.46\linewidth}
      \includegraphics[scale=0.5]{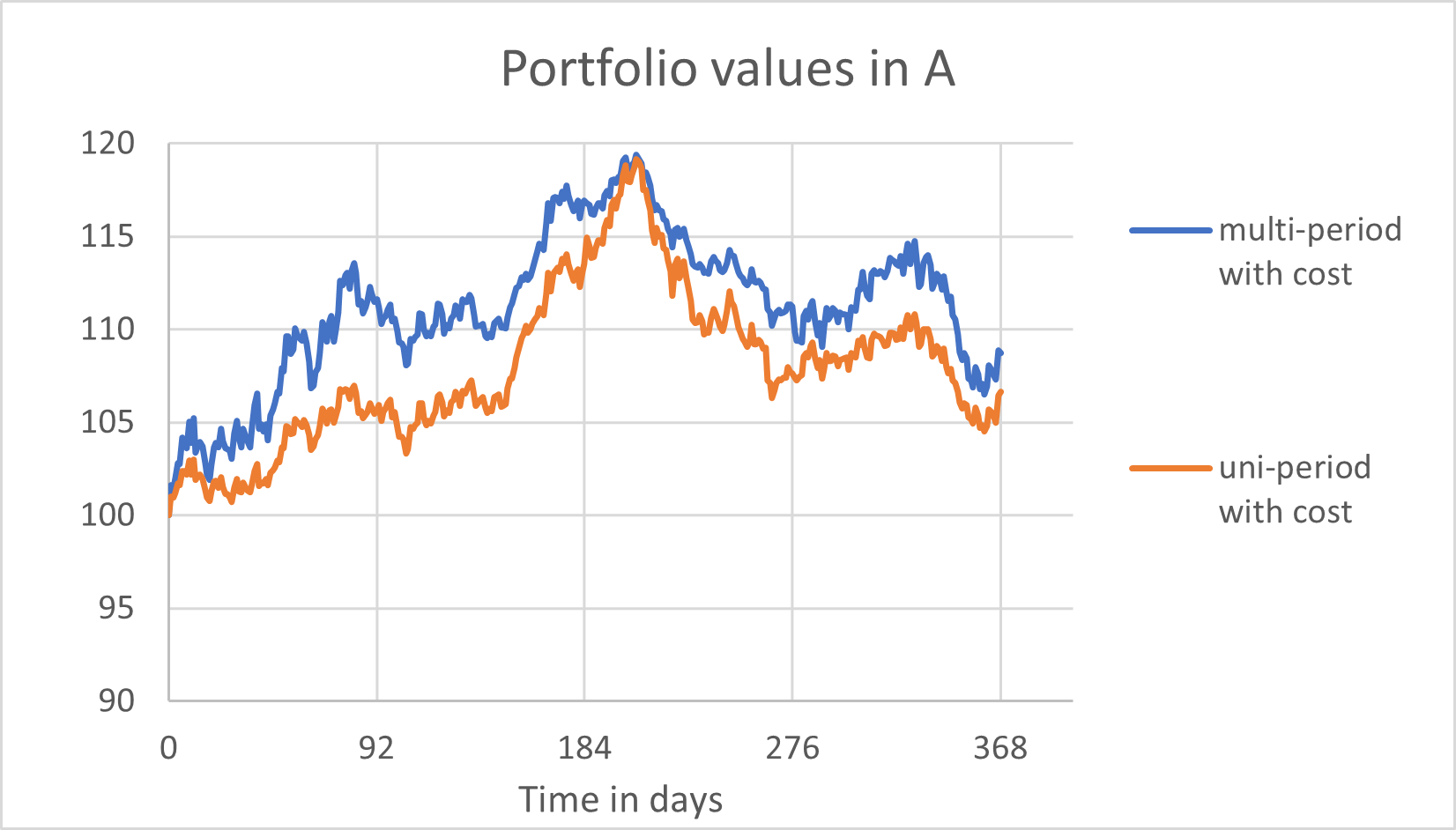} \caption{Portfolio values in A }
      \label{fig:portfolio-A-comp1}
   \end{minipage}\hfill
   \begin{minipage}[c]{0.46\linewidth}
      \includegraphics[scale=0.5]{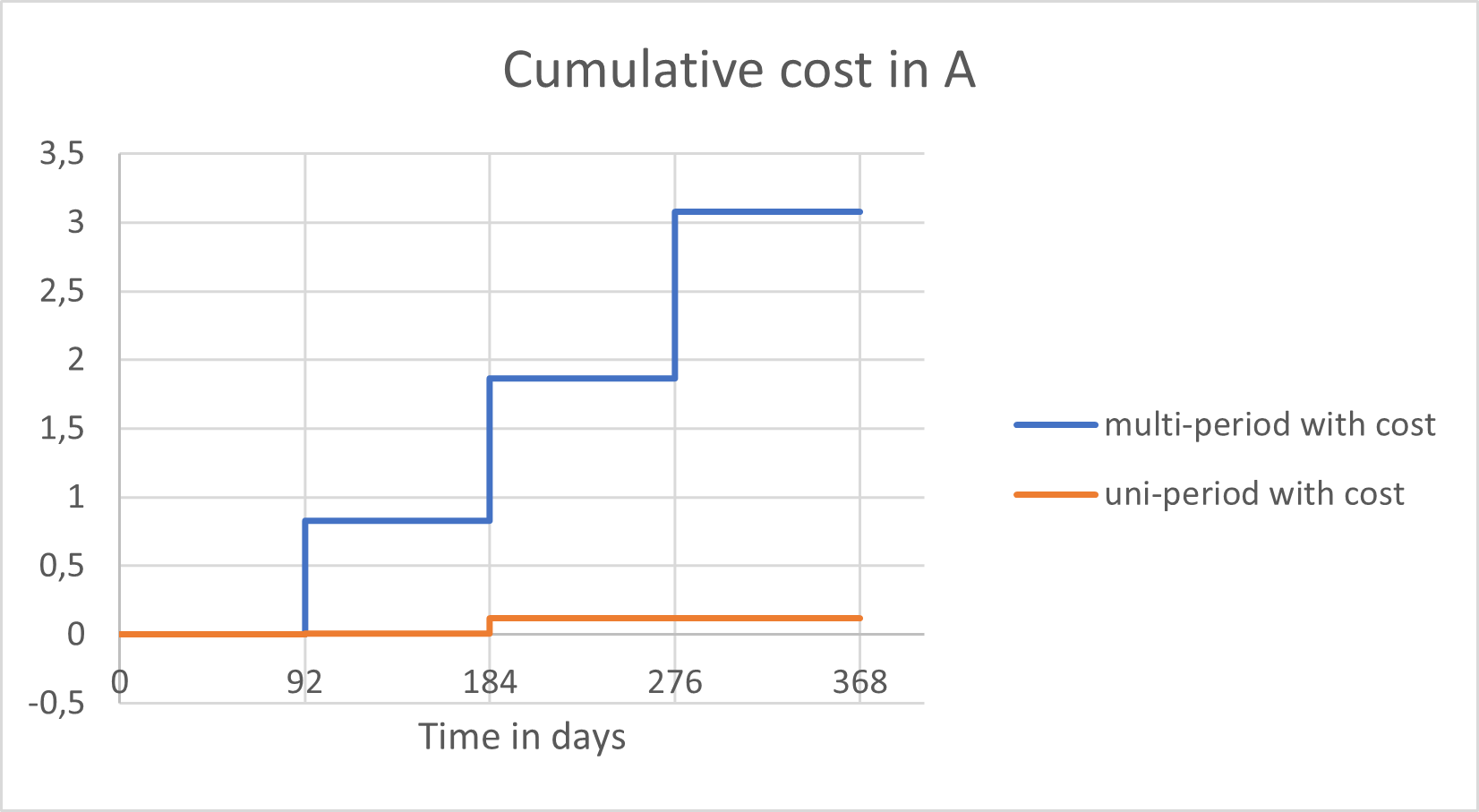}
      \caption{Cumulative cost in A}
      \label{fig:cost-A-comp1}
   \end{minipage} 
 \end{figure}
 
\begin{figure}[H]
   \begin{minipage}[c]{0.46\linewidth}
      \includegraphics[scale=0.5]{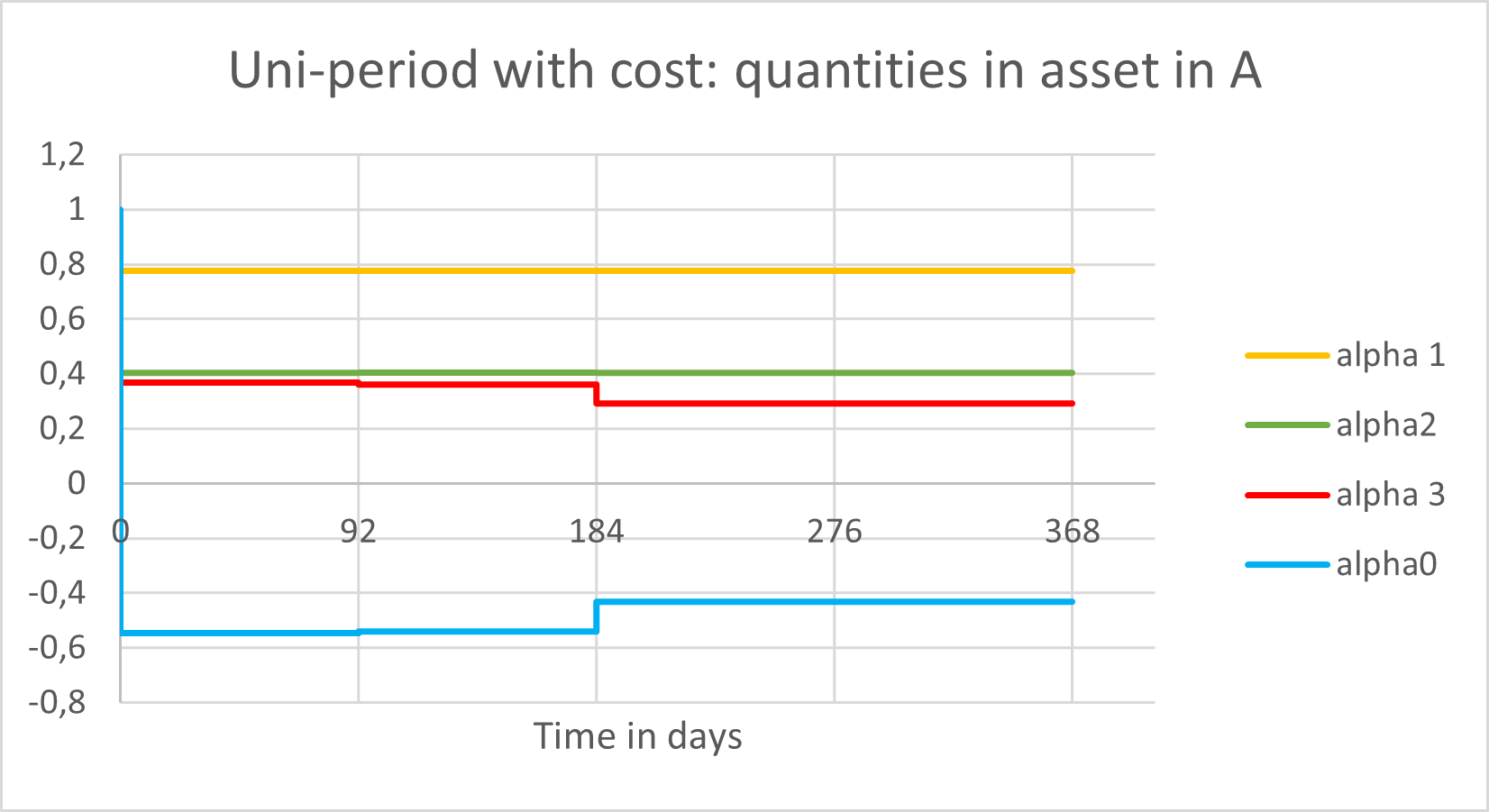}
      \caption{Controls of uni-period model with costs in A}
      \label{fig:control-A-uni-comp1}
   \end{minipage} \hfill
   \begin{minipage}[c]{0.46\linewidth}
      \includegraphics[scale=0.5]{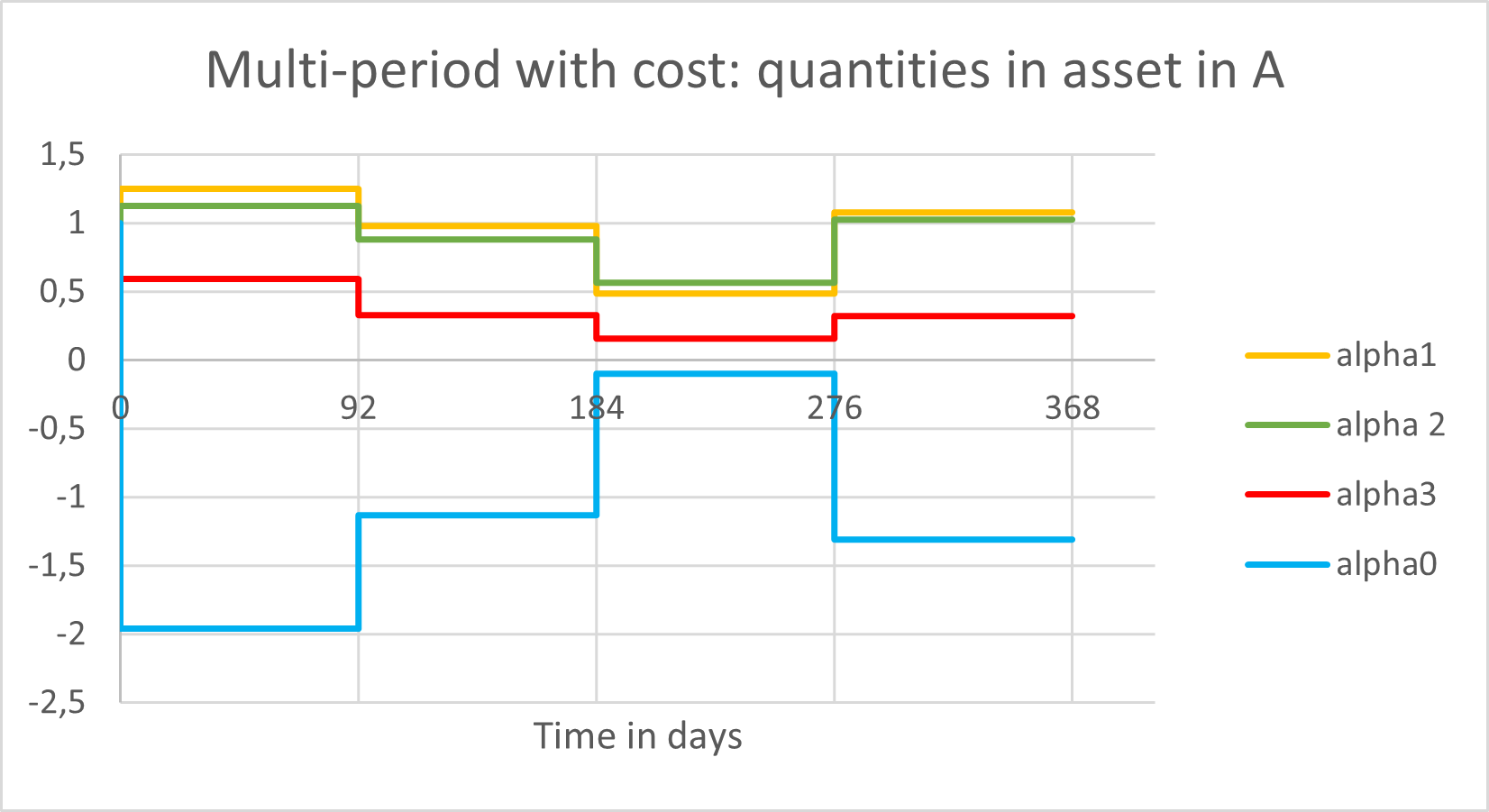} \caption{Controls of multi-period model with costs in A }
      \label{fig:control-A-multi-comp1}
   \end{minipage}
 \end{figure}

By analysing Figure \ref{fig:portfolio-A-comp1}, both portfolio trajectories follow the trend of asset 3. This behavior is not surprising given the flat evolution of assets 1 and 2. The multi-period portfolio outperforms at every moment, for this particular realisation.
\\
Figures \ref{fig:cost-A-comp1} highlights that the uni-period strategy pays few cost compared to the multi-period one. This can be attributed to the myopic vision of the uni-period strategy, which results in a very flat strategy that is not very sensitive to asset and wealth variations. It can not embrace the benefit of sacrificing money in paying costs to anticipate future evolution.   
The high costs required to make a reversal of strategy, further reinforce this inflexibility, which partly explains the low Sharpe ratio estimated in Section \ref{sec:result-with-cost}. The long position of this portfolio explains the high dependence with asset 3 and the decline of its value after the middle of the period.   
\\
The multi-period portfolio follows a completely different policy. While initially adopting a more aggressive long strategy than the uni-period portfolio, the level of risk taken is significantly higher. The portfolio performs well during the growth of asset 3 until day 276, at which point the strategy begins to reverse as assets are gradually sold. Subsequently, in response to the decline of asset 3, the strategy undergoes another shift, with new asset quantities being purchased.

\subsubsection{Behaviour on realisation B}

To illustrate the comparison, we present in Figures \ref{fig:asset-B}, a new realisation of the assets. This scenario looks like the previous one, as Asset 1 and 2 exhibit minimal changes while Asset 3 experiences a large increase without any significant decrease.
\begin{figure}[H]
\begin{center}
\includegraphics[scale=0.75]{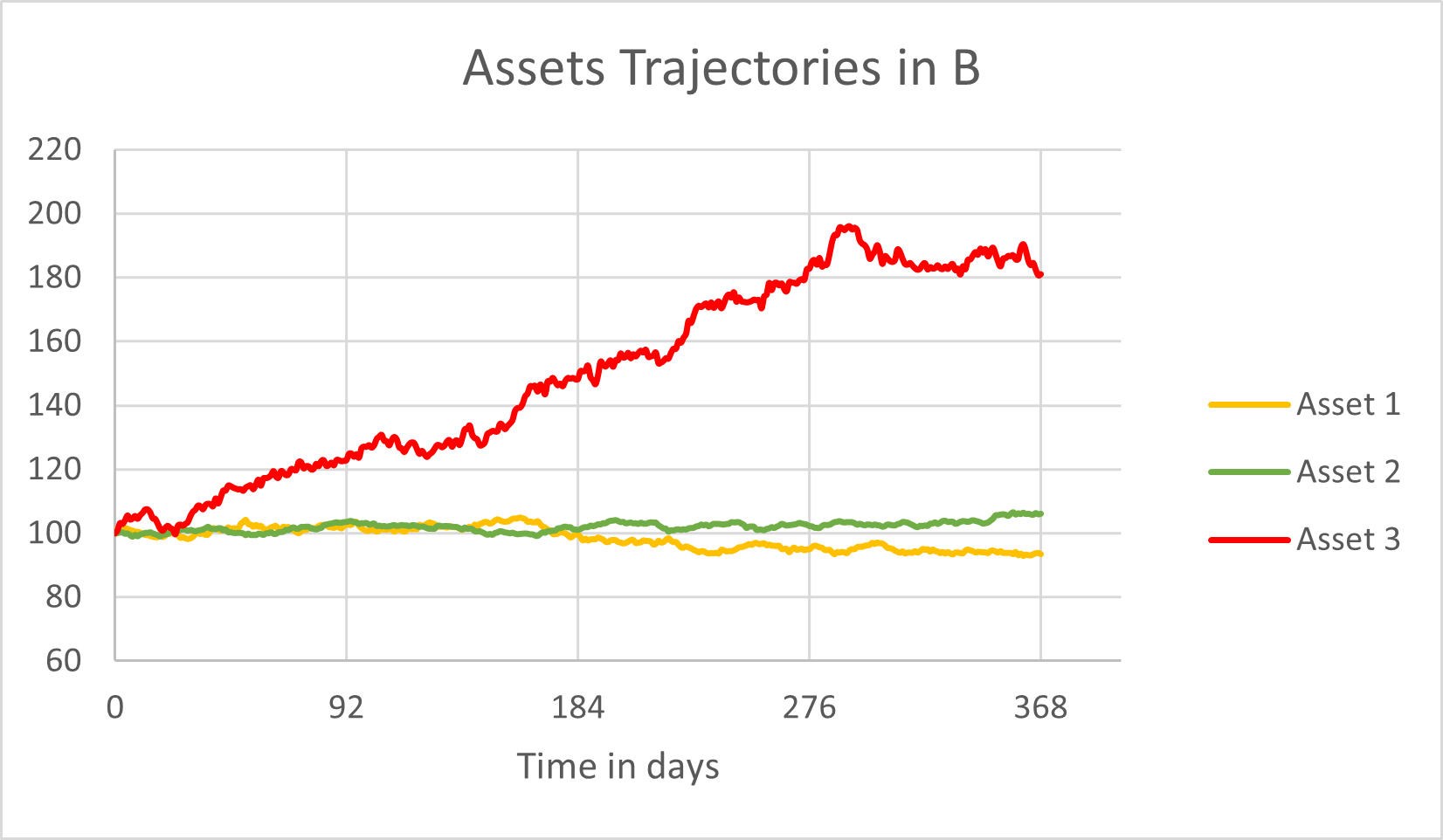}
      \caption{Asset trajectories in B}
      \label{fig:asset-B}
      \end{center}
\end{figure}

We present respectively in Figures \ref{fig:portfolio-B-comp1}, \ref{fig:cost-B-comp1}, \ref{fig:control-B-uni-comp1}, \ref{fig:control-B-multi-comp1}, the portfolios values, the cumulative costs and the controls of the uni-period model against the multi-period model with costs, whose performance results have been presented in Table \ref{tab:result-comp1}.

\begin{figure}[H]
\begin{minipage}[c]{0.46\linewidth}
      \includegraphics[scale=0.5]{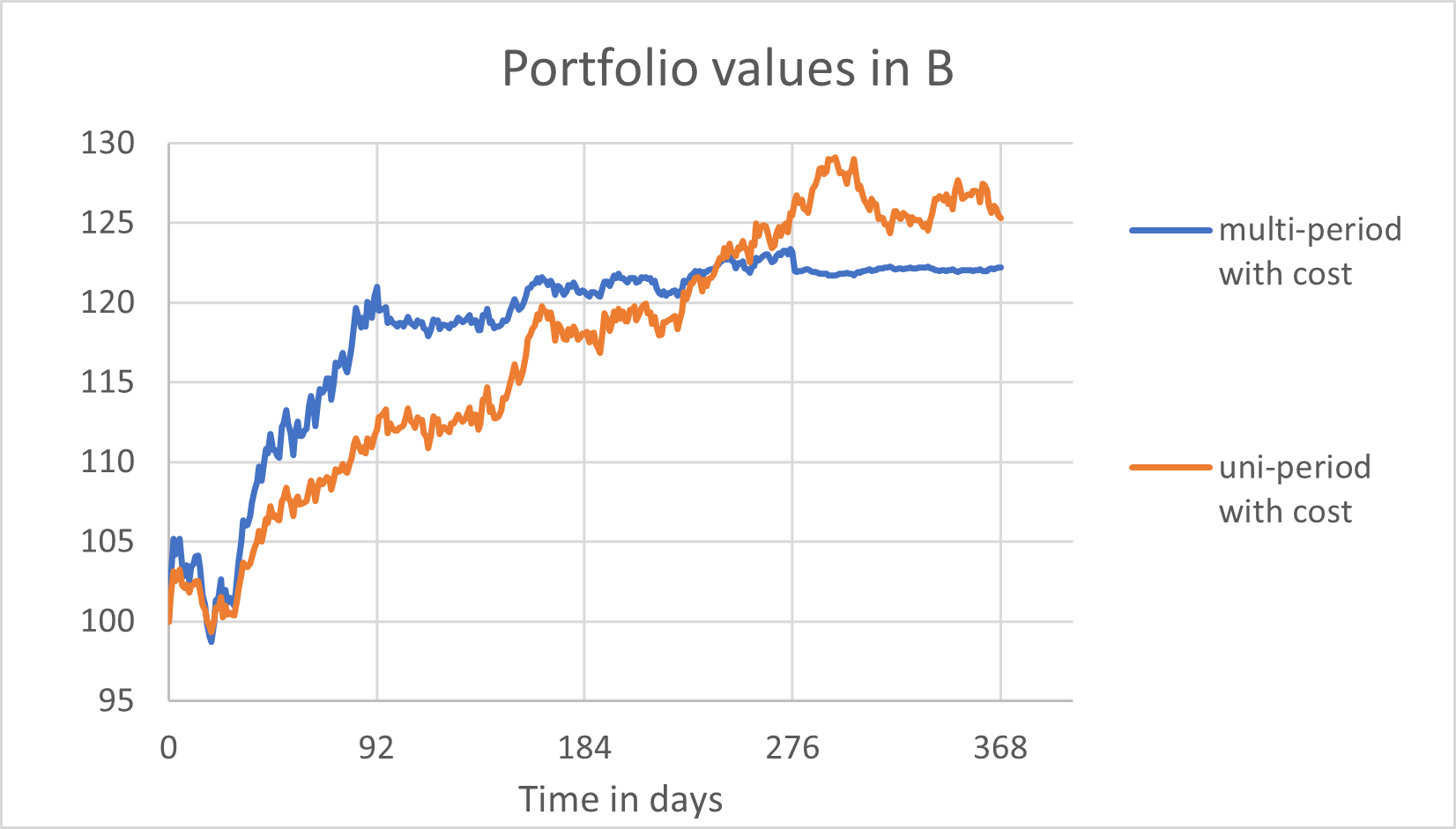} \caption{Portfolio values in B }
      \label{fig:portfolio-B-comp1}
   \end{minipage}\hfill
   \begin{minipage}[c]{0.46\linewidth}
      \includegraphics[scale=0.5]{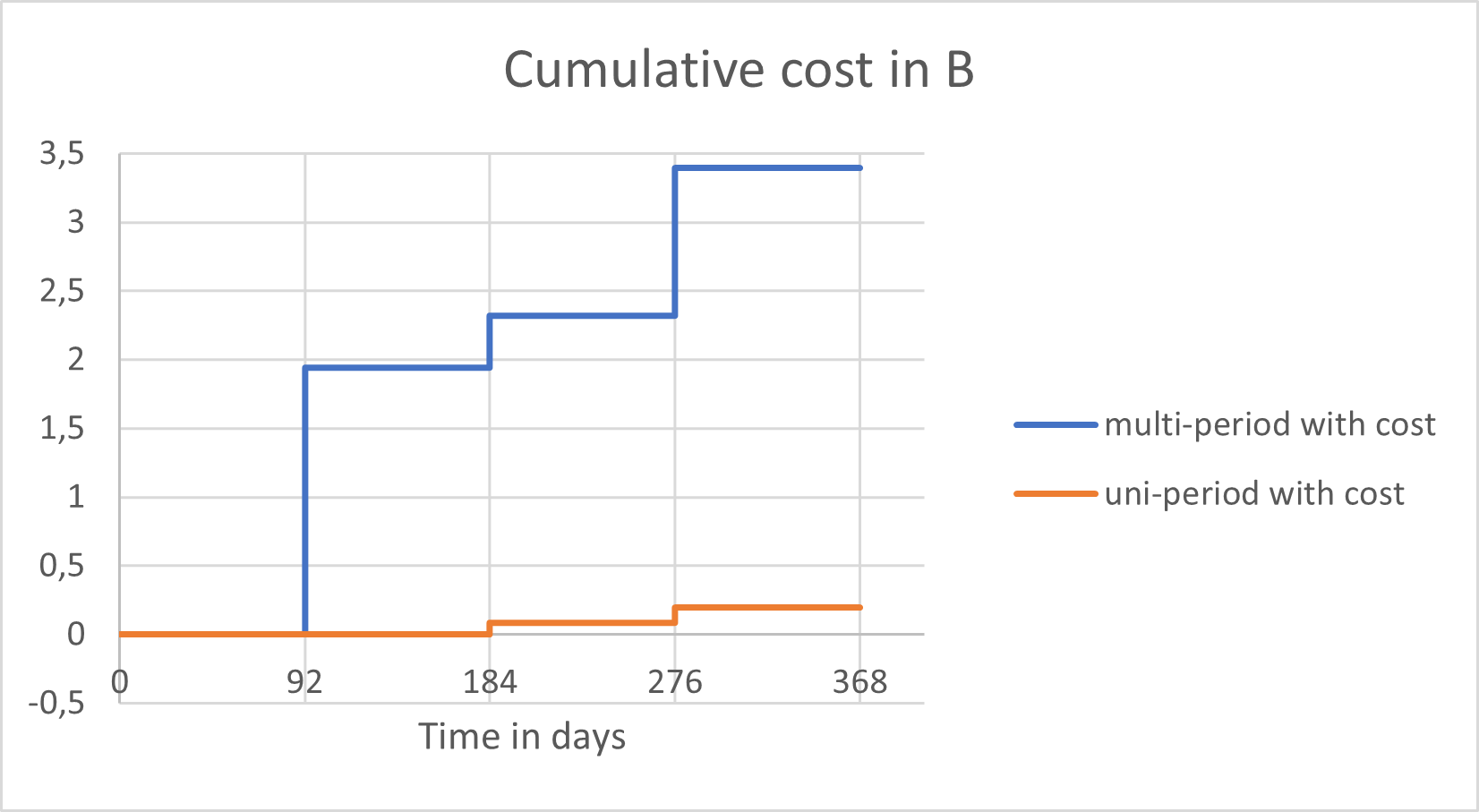}
      \caption{Cumulative cost in B}
      \label{fig:cost-B-comp1}
   \end{minipage} 
 \end{figure}
 
\begin{figure}[H]
   \begin{minipage}[c]{0.46\linewidth}
      \includegraphics[scale=0.5]{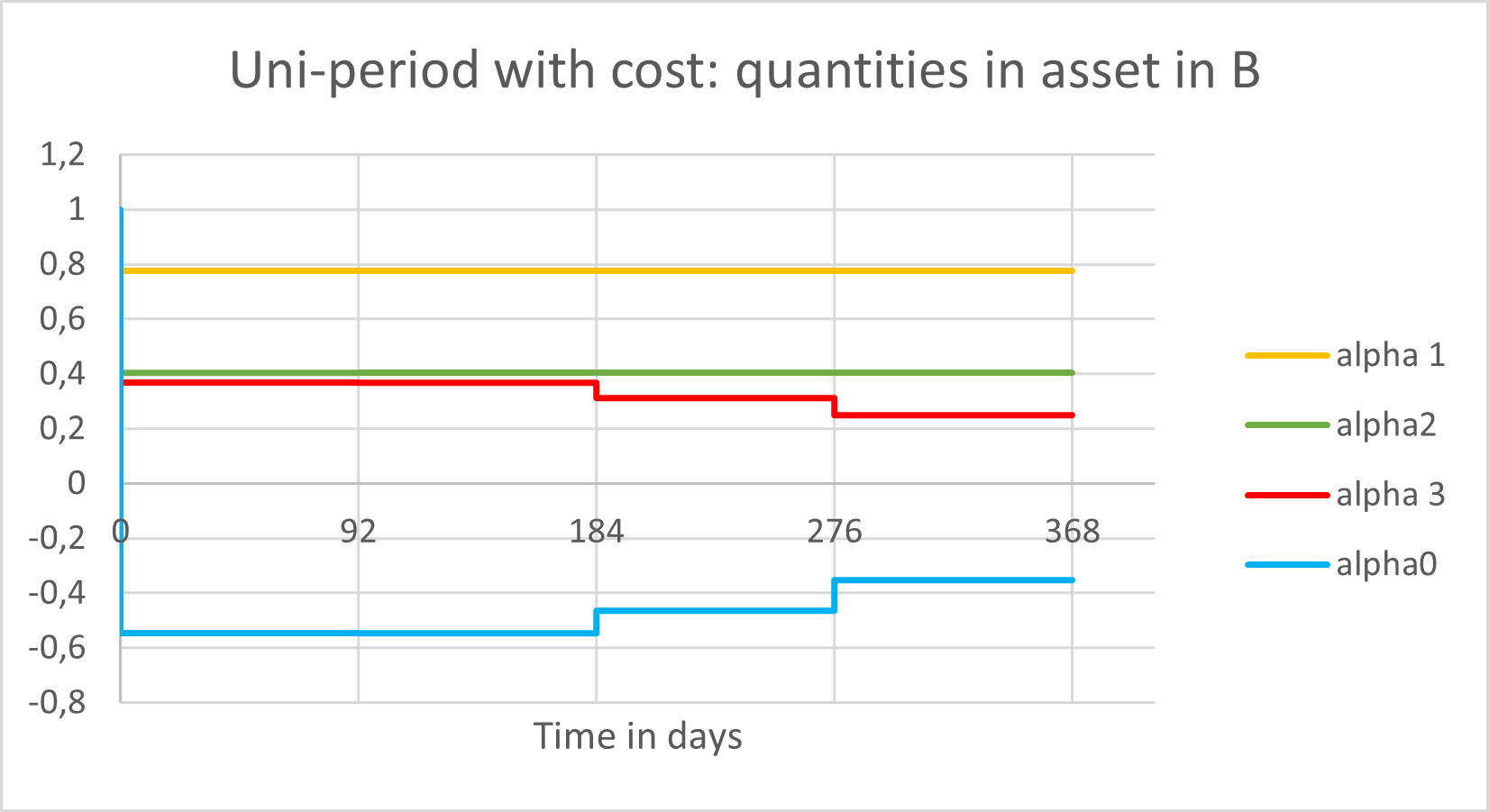}
      \caption{Controls of uni-period model with costs in B}
      \label{fig:control-B-uni-comp1}
   \end{minipage} \hfill
   \begin{minipage}[c]{0.46\linewidth}
      \includegraphics[scale=0.5]{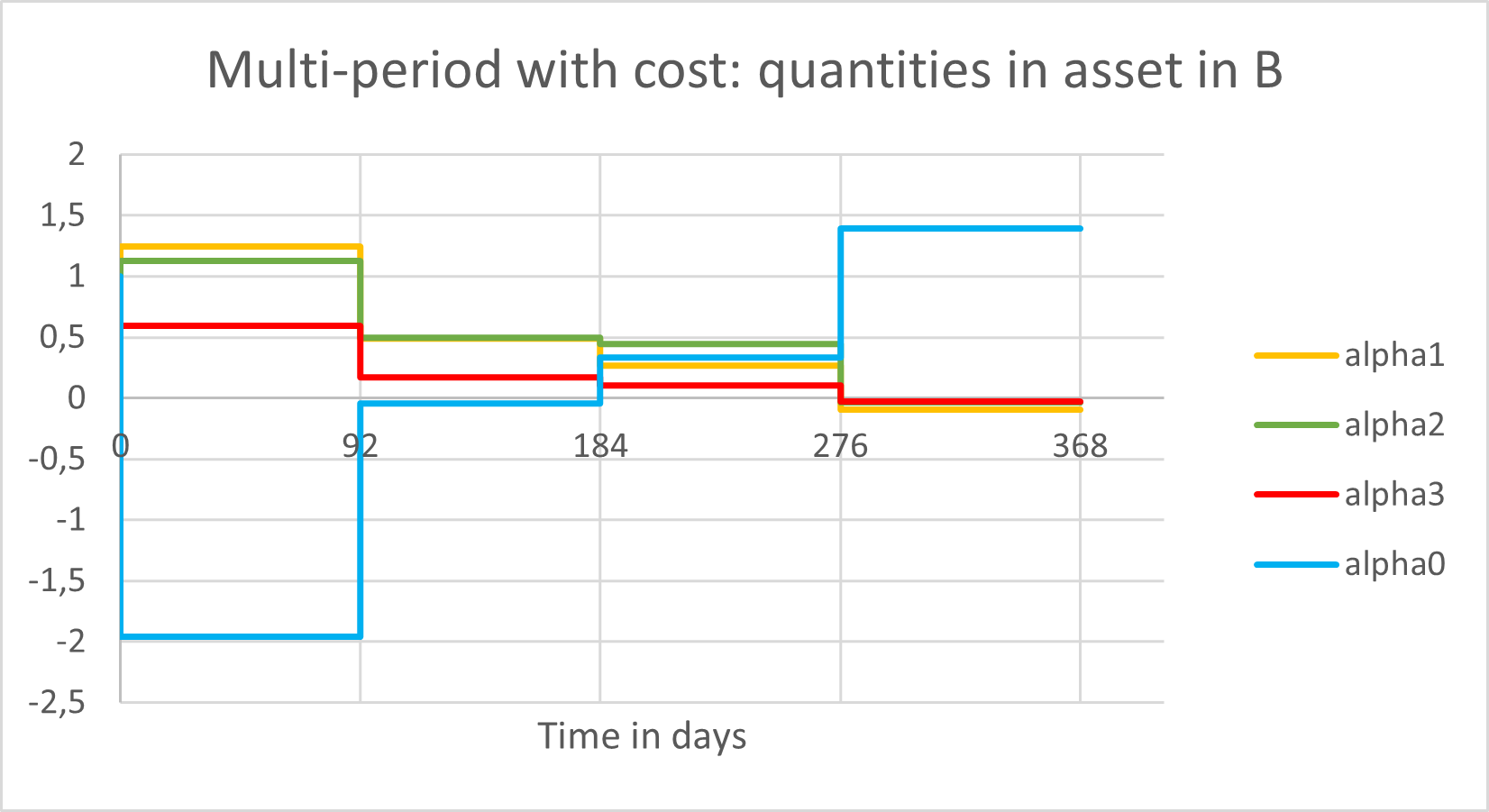} \caption{Controls of multi-period model with costs in B }
      \label{fig:control-B-multi-comp1}
   \end{minipage}
 \end{figure}

Although both portfolios experience growth over the period, the uni-period portfolio with costs outperforms the multi-period portfolio at the end. It is important to note that this outcome is specific to this particular trajectory but cannot be generalised. The uni-period portfolio appears to closely track the trend of asset 3, which continues to perform toward the end of the period. As for the realisation A, the positions of the uni-period strategy maintains predominantly long and flat positions. The multi-period strategy is considerably different and more versatile. Following 276 days of long positions, all assets are sold. This reverse of positions explains the plateau reached. After performing , the portfolio eliminates all risk to insure a positive performance. Even if Asset 3 were to continue its upward trend, this approach remains prudent. The portfolio achieves a higher value than anticipated, justifying its wish to protect its gains.

\section{Conclusion}
In this paper, we have presented an efficient numerical method to solve multi-stage portfolio allocation that involve multiple assets and transaction costs. We applied a stochastic descent gradient algorithm to find the Wiener chaos expansion of an optimal portfolio. The method could be extended to handle realistic constraints as no-shorting and is computationally tractable. 
We explored the link between risk aversion and optimal portfolios subject to transaction costs, with two findings standing out. Firstly, we have proved that the Sharpe ratio of a locally optimal portfolio with transaction costs does not depend on its risk aversion. Secondly, we have established that all optimal portfolios have the same Sharpe ratio. 
We have used this result to compare our approach to a competitive benchmark, based on the sequential uni-period mean-variance strategy.  We have highlighted the efficiency of our approach
and we have showcased our benchmark by analyzing the performance of our models on two selected trajectories. Costs must not be ignored in multi-period setting since reverses of strategy are frequent. On the other hand, considering costs for uni-period models is still a topic of debate. The myopic effect with the consideration of costs may freeze the strategy and negatively impact performance. As expected, the multi-period model is more intricate, but outperforms the uni-period models.




\appendix
\section{Wiener chaos expansion properties}
\label{sec:chaos}
In this Appendix, we present 
 several fundamental properties on Wiener chaos expansion and Hermite polynomials, use-full in our study. We refer to \cite{akahori2017discrete}, \cite{lelong2018dual}
\cite{nualart2006malliavin} for theoretical details.
\\
\\
Let $H_i$ be the $i$-th Hermite polynomial defined by 
\begin{equation*}
H_0(x)=1: \ H_i(x)=(-1)^i e^{\frac{x^2}{2}}\frac{d^i}{dx^i}(e^{\frac{-x^2}{2}}), \ \text{for} \ i\geq 1.
\end{equation*}
For $\lambda \in \mathbb{N}^n, \ x \ \in \mathbb{R}^n$, we define
$$
H_{\lambda}^{\bigotimes}(x)=\prod_{i=1}^n H_{\lambda_i}(x_i).
$$
We have the following properties
\begin{itemize}
\item
For $i \geq 1$, $H_i^{'}=iH_{i-1}$ with $H_{-1}=0$.
\item
For $x, y \in \mathbb{R}$,
\begin{equation*}
H_i(x+y)=\sum_{r=0}^i \binom{i}{r} x^r H_{i-r}(y).
\end{equation*}
\item
Let $X$, $Y$ be two random variables with joint Gaussian distribution such that $\E(X) = \E(Y ) = 0$ and $\E(X^2) = \E(Y^2) = 1$ . Then $\forall \ i, j\geq 0$, we have
\begin{equation*}
\E\left[H_i(X)H_{j}(Y)\right]=1_{(i=j)}i!(\E\left[XY\right])^i.
\end{equation*}
\item
Let $\lambda$ $\in \mathbb{R}^n \neq 0_{\mathbb{R}^n}$ and $X$ be a n-multivariate i.i.d standard normal vector. Then we have
\begin{equation*}
\E\left[H_\lambda^{\bigotimes}(X)\right]=0.
\end{equation*}
\item 
Let $\lambda$, $\lambda^{'}$ $\in \mathbb{R}^n$ and $X$ a n-multivariate i.i.d standard normal vector, then
\begin{equation*}
\E\left[H_\lambda^{\bigotimes}(X)H_{\lambda^{'}}^{\bigotimes}(X)\right]=1_{(\lambda=\lambda^{'})}\prod_{i=1}^n \lambda_{i} ! \ .
\end{equation*}
\end{itemize}
The following theorem represents the basics of our approach. It provides a discrete-time general representation with a convergent series expansion. 
This expansion is introduced and proved in \cite[Theorem 2.1]{akahori2017discrete}.
\begin{theorem}
\label{chaos}
Let Y  $\in \ L^2(\Omega,\cF_N,\mathbb{Q})$, then the following Wiener chaos expansion holds
\beaa
\begin{aligned}
Y & = \E_{\mathbb{Q}}[Y]+\sum_{\substack{\lambda \in (\mathbb{N}^{{N}})^d}} \beta_{\lambda} \prod_{j=1}^d\prod_{i= 1}^N H_{\lambda_{i}^j}\left(\frac{W^{\mathbb{Q}j}_{t_{i}}-W^{\mathbb{Q}j}_{t_{i-1}}}{\sqrt{t_{i}-t_{i-1}}}\right)\\
& = \E_{\mathbb{Q}}[Y]+\sum_{\substack{\lambda \in (\mathbb{N}^{{N}})^d}} \beta_{\lambda} H_{\lambda}^{\bigotimes}(\Delta \widehat{W}^{\mathbb{Q}}).
\end{aligned}
\eeaa

\end{theorem}
We call $\mathscr{C}_K(Y)$ the truncated expansion of order $K$
\[
\mathscr{C}_K(Y)=\E_{\mathbb{Q}}[Y]+\sum_{\substack{\lambda \in (\mathbb{N}^{N})^d \\ |\lambda|_1\leq K}} \beta_{\lambda} H_{\lambda}^{\bigotimes}(\Delta \widehat{W}^{\mathbb{Q}}).
\]
We have $\lim_{K \to +\infty} \E_{\mathbb{Q}}\left[|Y-\mathscr{C}_K(Y)|^2\right]=0$.
Let us define $m\in\mathbb{N}$ be the number of coefficients $\lambda$ appearing 
 in the chaos expansion. 
We have $m=\#\{\lambda \in (\mathbb{N}^{N})^d \,:\, |\lambda|_1\leq K\}=\binom{Nd+K}{Nd}$.
We slightly abuse the above definition and also write
$$
\mathscr{C}_K(\beta)=\E_{\mathbb{Q}}[Y]+\sum_{\substack{\lambda \in (\mathbb{N}^{N})^d \\ |\lambda|_1\leq K}} \beta_{\lambda}  H_{\lambda}^{\bigotimes}(\Delta \widehat{W}^{\mathbb{Q}}).
$$
\begin{proposition}
\label{prop:tronc}
Let $Y\in \ L^2(\Omega,\cF_N,\mathbb{Q})$, then 
\be
\label{eq:tronc}
\E_{\bQ}\left[\mathscr{C}_K(Y)|\cF_{n}\right]= \mathscr{C}_K(\E_{\bQ}\left[Y|\cF_{n}\right]) = \mathscr{C}_{K,n}(Y),
\ee
where
$$
\mathscr{C}_{K,n}(Y)=\E_{\mathbb{Q}}[Y]+\sum_{\substack{\lambda \in (\mathbb{N}^{N})^d \\ |\lambda|_1\leq K, \\  \forall j\leq d, \forall l \geq n+1: \lambda^j_{l}=0}} \beta_{\lambda} H_{\lambda}^{\bigotimes}(\Delta \widehat{W}^{\mathbb{Q}})=\E_{\mathbb{Q}}[Y]+\sum_{\substack{\lambda \in (\mathbb{N}^{n})^d \\ |\lambda|_1\leq K}} \beta_{\lambda} H_{\lambda}^{\bigotimes}(\Delta \widehat{W}^{\mathbb{Q}}).
$$
The Wiener chaos expansion of the conditional expectation of the variable Y is obtained by truncating the terms which are not $\cF_n$-measurable. 
\end{proposition}

\section{Differentiability}
\label{sec:differentiability}
In this appendix, we study the differentiability of $(\beta,\theta)\rightarrow\E_{\Prob}[F^{\gamma}(\beta,\theta)]=\E_{\Prob}\left[\mathcal{R}(\beta)S_N^0-\gamma\left(\left(\mathcal{R}(\beta)-\theta\right)S_N^0\right)^2\right]$. 
In our resolution framework, we assume that $\sigma$ is deterministic.



\begin{proposition}
\label{prop_R_differentiable}
$\mathcal{R}$ is almost surely differentiable on $\mathcal{Z}$.
\end{proposition}
\begin{proof}
Let us prove that the $\mathbb{Q}(\alpha_{n+1}^i=\alpha_{n}^i)=0$, $\forall \ 1\leq i\leq d,\ 0\leq n\leq N-1$. We recall that 
\beaa
\alpha_{n+1}(\beta) = \E_\bQ\left[\Delta \tilde{S}_{n+1} \Delta (\tilde{S}_{n+1})^T|\cF_{n}\right]^{-1}\E_\bQ\left[\mathscr{C}_K(\Delta {Z}_{n+1})\Delta \tilde{S}_{n+1}|\cF_{n}\right].
\eeaa   
Let $\Diag$ be the operator which associates to a vector $(d^1,\dots,d^k)$ the matrix $\begin{bmatrix}d_{1} & & \\ & \ddots & \\ & & d_{k}\end{bmatrix}$. 
\\
We define 
$$
Y_{n}=\E_\bQ\left[ \left(\frac{\tilde{S}_{n+1}}{\tilde{S}_n}-1\right)\left(\frac{\tilde{S}_{n+1}}{\tilde{S}_n}-1\right) ^T|\cF_{n}\right]=\left[e^{\sigma_n^k\cdot\sigma_n^j(t_{n+1}-t_n)}-1\right]_{k,j}.
$$
Note that $Y_n$ is deterministic and invertible.
We slightly abuse of the notation $\mathscr{C}_K$ to denote the vector $\E_{\bQ}\left[\mathscr{C}_K(\Delta {Z}_{n+1})\left( \frac{S_{n+1}}{S_n}-1\right)|\cF_{n}\right]=\left(E_{\bQ}\left[\mathscr{C}_K(\Delta {Z}_{n+1})\left( \frac{\tilde{S}_{n+1}^j}{\tilde{S}_{n}^j}-1\right)|\cF_{n}\right]\right)_{j\in\{1,\ldots,d\}}$.
We have
$$ \forall \beta\in\mathcal{Z}, \ 
\alpha_{n+1}(\beta)=\Diag(\tilde{S}_n)^{-1}Y_n^{-1}\E_{\bQ}\left[\mathscr{C}_K(\Delta {Z}_{n+1})\left( \frac{\tilde{S}_{n+1}}{\tilde{S}_{n}}-1\right)|\cF_{n}\right].
$$
The term $\E_\bQ\left[\mathscr{C}_K(\Delta {Z}_{n+1})\left( \frac{\tilde{S}_{n+1}^j}{\tilde{S}_{n}^j}-1\right)|\cF_{n}\right]$ can be expressed as a linear combination of the chaos coefficients of $Z_N$ such that
\be
\label{eq:expectation}
\E_\bQ\left[\mathscr{C}_K(\Delta {Z}_{n+1})\left( \frac{\tilde{S}_{n+1}^j}{\tilde{S}_{n}^j}-1\right)|\cF_{n}\right]=\sum_{\substack{\lambda \in (\mathbb{N}^{n+1})^d \\ |\lambda|_1\leq K, \\  \lambda_{n+1} \neq 0}} \beta_{\lambda} H_{\lambda_{1:n}}^{\bigotimes}(\Delta \widehat{W}_{1:n}^{\mathbb{Q}})\E_\bQ\left[H_{\lambda_{n+1}}(\Delta \widehat{W}_{n+1}^{\bQ})(\frac{\tilde{S}_{n+1}^j}{\tilde{S}_{n}^j}-1)\right].
\ee
For $i\in\{1,\ldots,d\}$, we obtain
$$
\bQ(\alpha_{n+1}^i(\beta)=0)\iff\bQ\left(\left[Y_n^{-1}\E_{\bQ}\left[\mathscr{C}_K(\Delta {Z}_{n+1})\left( \frac{\tilde{S}_{n+1}}{\tilde{S}_{n}}-1\right)|\cF_{n}\right]\right]_i=0\right).
$$
Note that $ \left[Y_n^{-1}\E_{\bQ}\left[\mathscr{C}_K(\Delta {Z}_{n+1})\left( \frac{\tilde{S}_{n+1}}{\tilde{S}_{n}}-1\right)|\cF_{n}\right]\right]_i$ 
can be written as
\be
\label{eq:pstar}
\sum_{j=1}^d (Y_n^{-1})_{ij}\sum_{k^{'}\leq K}\sum_{\substack{\lambda_{1:n} \in (\mathbb{N}^{n})^d \\ |\lambda_{1:n}|_1= k^{'}}}\left(\sum_{\substack{\lambda_{n+1} \in \mathbb{N}^d \\ 1\leq |\lambda_{n+1}|_1\leq K-k^{'}}}\beta_{\lambda}\E_\bQ\left[H_{\lambda_{n+1}}(\Delta \widehat{W}_{n+1}^{\bQ})(\frac{\tilde{S}_{n+1}^j}{\tilde{S}_{n}^j}-1)\right]\right)H_{\lambda_{1:n}}^{\bigotimes}(\Delta \widehat{W}_{1:n}^{\mathbb{Q}}).
\ee
With the definition of $\mathcal{Z}$, we deduce that $\forall \beta\in\mathcal{Z},\ \bQ(\alpha_{n+1}^i(\beta)=0)=0$.
\\
Let us introduce the following two spaces for $n\leq 1$,
$$
\mathcal{P}_n=\left\{\sum_{k=1}^K b_k H_k(\Delta W_n), \ (b_k)_k \ \cF_{n-1}-\text{measurable}, \ K\in\mathbb{N}^{*}, \ \exists k\in\{1,..K\}, \ \mathbb{Q}(b_k=0)=0\right\},
$$
$$
\mathcal{E}_n=\left\{\sum_{l=0}^L a_l e^{c_l\Delta W_n}, \ (a_l)_l, \ (c_l)_l  \ \cF_{n-1}-\text{measurable}, \ L\in\mathbb{N}^{*}, \ \exists l\in\{1,..L\}, \ \mathbb{Q}(a_l=0)=\mathbb{Q}(c_l=0)=0\right\}.
$$
We can deduce from (\ref{eq:pstar}) that $\exists \ (p_n,e_n)\in\mathcal{P}_n\times\mathcal{E}_n$ such that $\alpha_{n+1}^i(\beta)=p_{n}e_n$. 
Finally, we conclude that $\forall \beta\in \mathcal{Z}$, $\mathbb{Q}\left(\alpha_{n+1}^i(\beta)=\alpha_{n}^i(\beta)\right)=\E_{\mathbb{Q}}\left[\E_{\mathbb{Q}}\left[\mathds{1}_{[\alpha_{n+1}^i(\beta)=\alpha_{n}^i(\beta)]}|\cF_{n-1}\right]\right]=0$. 
\end{proof}

\begin{proposition}
\label{prop:differentiability}
The function $$(\beta, \theta)\longmapsto\E_{\Prob}[F^{\gamma}(\beta,\theta)]=\E_{\Prob}\left[\mathcal{R}(\beta)S_N^0-\gamma\left(\left(\mathcal{R}(\beta)-\theta\right)S_N^0\right)^2\right]$$ is differentiable on $\mathcal{Z}\times \mathbb{R}$.
\end{proposition}
\begin{proof}
Let $r>0,\ r_{\theta}>0$ and $\beta\in\mathcal{Z}$, $\theta\in\mathbb{R}$ such that $|\beta|\leq r$, $|\theta|\leq r_{\theta}$. 
According to Lemma \ref{prop:R}, we have
$$
|\mathcal{R}(\beta)| \leq V_0+|D| |\beta| +\sum_{i=1}^d\sum_{n=0}^{N-1}\nu|K_n^i||\beta|. 
$$
Since we have $\forall \ i \in \{1,\dots,d\}, \ n\in\{0,\dots,N-1\}, \ \forall \ p\geq 1$, $\E_{\Prob}\left[|K_n^i|^p\right] + \E_{\Prob}\left[|D|^p\right] <\infty$,  we have $\E_{\Prob}\left[\underset{|\beta|\leq r}{\sup} \left|\mathcal{R}(\beta)\right|\right]<\infty$ and $\E_{\Prob}\left[\underset{|\beta|\leq r}{\sup} \left|\mathcal{R}(\beta)\right|^2\right]<\infty$.
We bound
\beaa
\label{eq:differentiable}
\begin{aligned}
|F^{\gamma}(\beta,\theta)| & \leq |\mathcal{R}(\beta)| +  
2\gamma|\mathcal{R}(\beta)|^2+ 2\gamma\theta^2 \\
 & \leq \underset{|\beta|\leq r}{\sup} \left|\mathcal{R}(\beta)\right| +
 2\gamma\underset{|\beta|\leq r}{\sup} \left|\mathcal{R}(\beta) \right|^2+2\gamma r_{\theta}^2.
\end{aligned}
\eeaa 
So $\E_{\Prob}\left[|F^{\gamma}(\beta,\theta)|\right]<\infty$. According to Proposition \ref{prop_R_differentiable}, $\mathcal{R}$ is 
 a.s. differentiable on $\mathcal{Z}$, then $F$ is also a.s. differentiable on $\mathcal{Z}\times \mathbb{R}$. 
 We also have
 $$|\nabla \mathcal{R}(\beta)|=|D-\sum_{i=1}^d\sum_{n=0}^{N-1}\nu \times \sign(K_n^i\cdot\beta)K_n^i|\leq  |D|+\sum_{i=1}^d\sum_{n=0}^{N-1}\nu|K_n^i|.
 $$
Following the same arguments, we prove that $\forall j \in\{1,\dots,m\}$
\beaa
\begin{aligned}
|\nabla_{\beta_j} F^{\gamma}(\beta,\theta)| & \leq |\nabla \mathcal{R}(\beta)|_j + 2\gamma |\nabla \mathcal{R}(\beta)|_j|\mathcal{R}(\beta)-\theta|\\
& \leq |D_j|+\sum_{i=1}^d\sum_{n=0}^{N-1}\nu|(K_n^i)_j|+2\gamma \left(|D_j|+\sum_{i=1}^d\sum_{n=0}^{N-1}\nu|(K_n^i)_j|\right)\left(\underset{|\beta|\leq r}{\sup} \left|\mathcal{R}(\beta)\right| +r_{\theta}\right).
\end{aligned}
\eeaa
We conclude using Cauchy Schwartz' inequality that $\E_{\Prob}[\sup_{|\beta| \le r, |\theta|\le r_\theta} |\nabla_{\beta_j} F^{\gamma}(\beta,\theta)|]<\infty$. Similarly, we bound 
\beaa
\begin{aligned}
|\nabla_{\theta} F^{\gamma}(\beta,\theta)| & \leq 2\gamma |\mathcal{R}(\beta)-\theta|\leq 2\gamma \underset{|\beta|\leq r}{\sup} \left|\mathcal{R}(\beta)\right|+r_{\theta}.
\end{aligned}
\eeaa
Finally, we apply Lebesgue's theorem to conclude that $(\beta,\theta)\rightarrow\E_{\Prob}[F^{\gamma}(\beta,\theta)]$ is differentiable on $\mathcal{Z}\times\mathbb{R}$ and $\nabla \left(\E_{\Prob}\left[F^{\gamma}(\beta,\theta)\right]\right)=\E_{\Prob}\left[\nabla F^{\gamma}(\beta,\theta)\right].$
\end{proof}

\section{Benchmark models}
\label{sec:benchmark}
Here, we present different approaches used to benchmark and challenge the multi-period one, described in Section \ref{sec:optimal-solution}. A first element of comparison is the multi-period approach itself but which ignores transaction costs. A second comparison is done with the sequential uni-period Markowitz framework. We present two versions of this approach, when cost are ignored and considered. We also study the link of risk aversion on optimal solutions for those benchmark models. 
\subsection{Benchmark: Multi-period allocation Ignoring transaction costs}
\label{sec:chaos-ignoring-cost}
A natural an simplified approach consists in applying the framework described in Section \ref{sec:optimal-solution}, while ignoring costs. Costs are removed from the portfolio value but do not impact the strategy. Theoretically, we look for an optimal strategy whose costs are refunded. The control $(\alpha_n^{\star})_n$, are solution to \eqref{sys:cost2} by temporally setting $\nu=0$ during the resolution. 
Formally, we solve
\be
\label{system_ignored-cost}
\tag{$E^{\gamma}_0$}
\begin{aligned}
& \sup_{(\alpha_n)_{n\in \{1,N\}}}
& & \E_{\Prob}\left[U(\tilde{X}_{N}S_N^0)\right] \\
& \text{s.t.}
& & \tilde{X}_0=V_0, \ (\alpha_n)_n \ \bfF-Pred\\
& & &  \tilde{X}_{n+1}=\tilde{X}_{n}+\alpha_{n+1}\cdot\Delta \tilde{S}_{n+1}
\end{aligned}
\ee
Then, the real quantity in risk free asset at stake is deduced with the auto-financing relation and by removing costs generated by the real value of $\nu$. Formally, after maximizing \eqref{system_ignored-cost}, the generated cost $\mathcal{C}_N$ is computed and removed to obtain the real portfolio value $\tilde{V}_N=\tilde{X}_N-\mathcal{C}_N$. This strategy is implemented in order to observe the effect of costs on optimal strategies.

\subsection{Benchmark: Sequential Uni-period Markowitz portfolio allocation}
\label{sec:uni}
In this section, we present an alternative approach, intended to serve as a benchmark. 
We would like to compare the performance of our method to a method traditionally applied by asset managers. We propose to implement the sequential uni-period mean–variance Markowitz framework. This approach also corresponds to the one-time-step Model Predictive Control (MPC) described in \cite{li2022multi}. 
In order to maximise the rate of return of a portfolio while controlling its volatility at time $T$, an agent consecutively maximizes at time $0=t_0<\cdots<t_N=T$, its uni-period mean-var objective function. We assume that the agent has an uni-period risk aversion parameter $\gamma_u$. The agent has to consecutively solve for $n \in \{0,\cdots,N-1\}$,
\be
\label{system_uniperiod_ap2}
\tag{$Y^{\gamma_u}_{n}$}
\begin{aligned}
\ & \sup_{\alpha_{n+1}\in\mathbb{R}^d, \alpha_{n+1}^0\in\mathbb{R}}
& & \E_{\Prob}\left[{V}_{n+1}|\cF_{n}\right]-\gamma_u \Var_\bP\left[V_{n+1}|\cF_{n}\right] \\
& \text{s.t.}
& & \Delta V_{n+1}=\alpha_{n+1}\cdot \Delta S_{n+1}+\alpha_{n+1}^0 \Delta S_{n+1}^0 -\sum_{i=1}^d\nu|\alpha_{n+1}^i-\alpha_{n}^i|S_{n}^i
\end{aligned}
\ee
We present two versions of this model. A naive version, where the asset manager pays the cost but does not take the amount into account in its strategy (by ignoring them), is presented in Section \ref{sec:uni-ignoring}. In a second version, described in Section \ref{seq:minVarCost}, the investor consecutively solves~\eqref{system_uniperiod_ap2} and bases his strategy according to the cost he will have to pay. 

\subsubsection{Ignoring costs}
\label{sec:uni-ignoring}
We aim to propose here, a simplified version of the initial sequential uni-period problem~\eqref{system_uniperiod_ap2}. The underlying idea is the same as the multi-period version described in Section~\ref{sec:chaos-ignoring-cost}. At each time step, the agent maximizes its mean-var utility function as though costs are refunded. Costs do not impact the strategy and are just removed from the portfolio value at the end. 
The agent has to consecutively solve for $n \in \{0,\dots,N-1\}$,
\be
\label{system_uniperiod_0}
\tag{$\mathcal{Y}^{\gamma_{u}}_{n}$}
\begin{aligned}
\ & \sup_{\alpha_{n+1}\in\mathbb{R}^d,\alpha_{n+1}^0\in\mathbb{R}}
& & \E_{\Prob}\left[{X}_{n+1}|\cF_{n}\right]-\gamma_u \Var_\bP\left[X_{n+1}|\cF_{n}\right] \\
& \text{s.t.}
& & \Delta X_{n+1}=\alpha_{n+1}\cdot \Delta S_{n+1}+\alpha_{n+1}^0 \Delta S_{n}^0 
\end{aligned}
\ee
Let $(\alpha_{n+1}^{\star},\alpha_{n+1}^{0 \ \star})$ be a solution and $(X_n^{\star})$, the associated portfolio. 
The real portfolio is computed after removing the generated costs such that  $V_{n+1}^{\star}=X_{n+1}^{\star}-\mathcal{C}_{n+1} S_{n+1}^0.$
According to the dynamics of the assets, we have for $n\in\{0,\dots,N-1\}$, 
$$
\left\{
    \begin{array}{ll}
    \E_{\Prob}\left[S_{n+1}^i|\cF_{n}\right]=S_{n}^i e^{\mu_n^i(t_{n+1}-t_n)}, \ \forall \ 1\leq i\leq d \ , \\
        \Cov\left[S_{n+1}^i,S_{n+1}^j|\cF_{n}\right]=S_{n}^i S_{n}^j e^{(\mu_n^i+\mu_n^j)(t_{n+1}-t_n)}(e^{(\sigma_n^i)^T\sigma_n^j(t_{n+1}-t_n)}-1), \ \forall \ 1\leq i,j\leq d \ .
    \end{array}
\right.
$$
By calling $A_n=\Diag( (e^{\mu_n^i(t_{n+1}-t_n)})_{1 \le i \le d})$
and $B_n=[e^{(\mu_n^i+\mu_n^j)(t_{n+1}-t_n)}(e^{(\sigma_n^i)^T\sigma_n^j(t_{n+1}-t_n)}-1)]_{ij}$, 
We have
\beaa
\begin{aligned}
& \alpha_{n+1}^{\star}= \arg\sup_{\alpha_{n+1}\in\mathbb{R}^d}
& &  \alpha_{n+1}^T A_n S_{n} +(X_{n}-\alpha_{n+1}^T S_{n})\frac{S_{n+1}^0}{S_{n}^0} -\gamma_u \alpha_{n+1}^{T}S_{n}^T B_n S_{n}\alpha_{n+1}\\
\end{aligned}
\eeaa
\be
\begin{aligned}
\label{eq:uni-withoutcost}
  & \Leftrightarrow A_n S_{n}-S_{n}\frac{S_{n+1}^0}{S_{n}^0}-2\gamma_u(S_{n}^T B_n S_{n})\alpha_{n+1}^{\star}=0
& \Leftrightarrow \alpha_{n+1}^{\star}=\frac{A_nS_{n}-S_{n}e^{r_n(t_{n+1}-t_n)}}{2\gamma_u S_{n}^T B_n S_{n}}.
 \end{aligned}
\ee
\begin{remark}
It is equivalent to remove costs at each time step or to remove them at the end because the quantity in the risk free asset does not impact the strategy. 
\end{remark}
Comparing the multi-period strategy with this framework is a difficult task. The risk aversion parameters $\gamma$ and $\gamma_u$ in the two models, do not refer to the same risk aversion for the agent. We need to find a correct matching, or a measure of performance, independent from risk aversion. Obviously, this measure is the Sharpe ratio. 
\begin{proposition}
\label{prop:uni-sharp}
The Sharpe ratio of the sequential uni-period Markowitz strategy which ignores costs, does not depend on the risk aversion parameter $\gamma_u$. 
\end{proposition}
\begin{proof}
We prove by recurrence $\forall \ 0\leq n\leq N$, $\mathcal{H}_n: \exists\  \text{a $\cF_n-$measurable random variable}\ \mathcal{X}_n \  $ \\
$\text{not function of} \ \gamma_u \ \text{such that} \ 
V_n^{\star}=V_0S_n^0+\frac{\mathcal{X}_n}{2\gamma_u}.$
\\
$\mathcal{H}_0$ is true with $\mathcal{X}_0=0$. We assume $\mathcal{H}_n$ true for $n>0$. We have
\beaa
\begin{aligned}
\tilde{V}_{n+1}^{\star}
&= \tilde{V}_n^{\star}+\alpha_{n+1}^{\star}\cdot\Delta \tilde{S}_{n+1} -\sum_{i=1}^d \nu\frac{|\alpha_{n+1}^{i \ \star}-\alpha_n^{i \ \star}|S_{n}^i}{S_{n}^0}\\
&=\frac{1}{S_{n}^0}\left(V_0S_n^0+\frac{\mathcal{X}_n}{2\gamma_u}\right) +(\alpha_{n+1}^{\star})^T\Delta \tilde{S}_{n+1}-\sum_{i=1}^d \nu\frac{|\alpha_{n+1}^{i \ \star}-\alpha_n^{i \ \star}|S_{n}^i}{S_{n}^0}.\\
\end{aligned}
\eeaa
By using that $\tilde{Rd}_n=\Diag(\frac{\tilde{S}_{n}^i}{\tilde{S}_{n-1}^i})$, and the form of the solution in (\ref{eq:uni-withoutcost}), we rewrite
\beaa
\begin{aligned}
(\alpha_{n+1}^{\star})^T\Delta \tilde{S}_{n+1}
& =(\alpha_{n+1}^{\star})^T(\tilde{Rd}_{n+1}-Id)\tilde{S}_n = \tilde{S}_n^T(\tilde{Rd}_{n+1}-Id)\alpha_{n+1}^{\star}\\
& = \frac{1}{S_n^0}\frac{S_n^T(\tilde{Rd}_{n+1}-Id)(A_nS_n-S_ne^{r_n(t_{n+1}-t_n)})}{2\gamma_u S_n^T B_n S_n}.
\end{aligned}
\eeaa
\beaa
\begin{aligned}
\text{and} \ \sum_{i=1}^d \nu\left(\frac{|\alpha_{n+1}^{i \ \star}-\alpha_n^{i \ \star}|S_{n}^i}{S_{n}^0}\right)=\sum_{i=1}^d {\nu}\left|\frac{A_n^iS_{n}^i-S_{n}^ie^{r_n(t_{n+1}-t_n)}}{2\gamma_u S_{n}^T B_n S_{n}}-\frac{A_{n-1}^iS_{n-1}^i-S_{n-1}^ie^{r_{n-1}(t_{n}-t_{n-1})}}{2\gamma_u S_{n-1}^T B_{n-1} S_{n-1}}\right|\frac{S_{n}^i}{S_{n}^0}.
\end{aligned}
\eeaa

We prove $\mathcal{H}_{n+1}$ by denoting 
\beaa
\begin{aligned}
\mathcal{X}_{n+1} &=\frac{S_{n+1}^0}{S_{n}^0}\left(\mathcal{X}_n +\frac{S_n^T(Rd_{n+1}-Id)(A_nS_n-S_n e^{r_n(t_{n+1}-t_n)})}{ S_n^T B_n S_n}\right)\\
& -\frac{S_{n+1}^0}{S_{n}^0}\sum_{i=1}^d {\nu}\left|\frac{A_n^iS_{n}^i-S_{n}^ie^{r_n(t_{n+1}-t_n)}}{ S_{n}^T B_n S_{n}}-\frac{A_{n-1}^iS_{n-1}^i-S_{n-1}^ie^{r_{n-1}(t_{n}-t_{n-1})}}{ S_{n-1}^T B_{n-1} S_{n-1}}\right|S_{n}^i.
\end{aligned}
\eeaa
Using this form we deduce that $\E_{\Prob}[V_N^{\star}]-V_0S_N^0=\frac{\E_{\Prob}[\mathcal{X}_N]}{2\gamma_u}, \ \text{and} \ 
\Var[V_N^{\star}]=\frac{\Var[\mathcal{X}_N]}{4\gamma_u^2}.$ Finally
$$
\sharpe\left[V_N^{\star}\right]=\frac{\frac{\E_{\Prob}[V_N^{\star}]-V_0}{V_0}-\frac{V_0S_N^0-V_0}{V_0}}{\Var[\frac{V_N^{\star}-V_0}{V_0}]^{\frac{1}{2}}}=\frac{\E_{\Prob}[\mathcal{X}_N]}{\Var[\mathcal{X}_N]^{\frac{1}{2}}},
$$
which does not depend on $\gamma_u$.
\end{proof}
\subsubsection{Taking costs into account}
\label{seq:minVarCost}
In this version, costs are taken into account in the objective function but the strategy remains myopic. 
We recall that the agent has to consecutively solve \eqref{system_uniperiod_ap2}. 
In the presence of costs, we do not provide an explicit solution of \eqref{system_uniperiod_ap2}.
\\
\\
In order to use this model as a benchmark, let us prove an analog result to Proposition \ref{prop:uni-sharp}.
\begin{proposition}
\label{prop:uni-sharp-cost}
The Sharpe ratio of the sequential uni-period Markowitz strategy considering costs, does not depend on the risk aversion parameter $\gamma_u$. 
\end{proposition}
\begin{proof}
With the notation of the previous section, \eqref{system_uniperiod_ap2} can be rewritten as 
\beaa
\begin{aligned}
\ & \sup_{\alpha_{n+1}\in\mathbb{R}^d}
& & \alpha_{n+1}^T A_n S_{n} +\left(V_{n}-\alpha_{n+1}^T S_{n}-\sum_{i=1}^d-\nu(|\alpha_{n+1}^i-\alpha_{n}^i|S_{n}^i)\right)\frac{S_{n+1}^0}{S_{n}^0} -\gamma_u \alpha_{n+1}^{T}S_{n}^T B_n S_{n}\alpha_{n+1} \\
& \text{subject to}
& & V_{n+1}=\alpha_{n+1} S_{n+1}+\alpha_{n+1}^0 S_{n+1}^0
\end{aligned}
\eeaa
Let us consider the objective functions 
$$\forall\ n\leq N-1, \ f_n(\alpha)=\alpha^T A_n S_{n} +\left(V_{n}-\alpha^T S_{n}-\sum_{i=1}^d-\nu(|\alpha^i-\alpha_{n}^i|S_{n}^i)\right)\frac{S_{n+1}^0}{S_{n}^0} -\gamma_u \alpha^{T}S_{n}^T B_n S_{n}\alpha.
$$
$f_n$ is strictly concave and $\lim\limits_{\substack{|\alpha| \to +\infty}} f_n(\alpha)=-\infty$, then \eqref{system_uniperiod_ap2} has an unique solution. We call $\alpha_{n+1}^{\star}$ the solution of \eqref{system_uniperiod_ap2}. 
$f_n$ is differentiable on $\mathbb{R}^d\setminus \mathcal{O}_{n}$, where $\mathcal{O}_{n}=\{\alpha\in \mathbb{R}^d, \exists i \in\{1,\dots,d\} \ \text{such that}\ \alpha^i=\alpha_{n}^i \}$. $f_n$ admits a sub-differential $\partial f_n$, at any point $\alpha\in \mathbb{R}^d$
$$
\partial f_n(\alpha)= \left\{
    \begin{array}{ll}
        A_n S_{n}-2\gamma_u(S_{n}^T B_n S_{n})\alpha-\Diag(1+\nu \epsilon)S_{n}\frac{S_{n+1}^0}{S_{n}^0}, \
       \text{with} \\ \epsilon^i=\sign(\alpha^i-\alpha_{n}^i)  \ \text{if} \ \alpha^i\neq\alpha_{n}^i, \ \epsilon^i \in [-1,1] \ \text{otherwise}
    \end{array}
\right\}.
$$
We have $0\in \partial f_n( \alpha_{n+1}^{\star})$. 
Let us show by recurrence $\forall \ 0\leq n\leq N$, $\mathcal{H}_n: \exists \ \text{a $\cF_n-$measurable random variable} \  \ \mathcal{X}_n$ \\
$\text{not function of} \ \gamma_u \ \text{such that} \ 
\alpha_n^{\star}=\frac{\mathcal{X}_n}{2\gamma_u}.$
\\
$\mathcal{H}_0$ is true with $\mathcal{X}_0=0$. 
We assume $\mathcal{H}_n$ true for $n>1$. If $\exists \ i \in\{1,\dots,d\}$ such that ${\alpha}_{n+1}^{i\star}={\alpha}_{n}^{i\star}$ then it is sufficient to choose $\mathcal{X}_{n+1}^i=\mathcal{X}_{n}^i$. 
If ${\alpha}_{n+1}^{i\star}\neq{\alpha}_{n}^{i\star}$, then
$$
\alpha_{n+1}^{i\star}=\frac{A_n^{ii}S_{n}^i-(1-\nu \sign(\alpha_{n+1}^{i\star}-\alpha_{n}^{i\star}))S_{n}^ie^{r}}{2\gamma_u S_{n}^T B_n S_{n}}.
$$
Let us define $\mathcal{Z}$ a $\cF_n$-measurable random variable such that $\alpha_{n+1}^i=\frac{\mathcal{Z}}{2\gamma_u}$. Then with the form of $\alpha_n^{i\star}$, we have
$$
\mathcal{Z}=\frac{A_n^{ii}S_{n}^i-(1-\nu \sign(\mathcal{Z}-\mathcal{X}_n^i))S_{n}^ie^{r}} {S_{n}^T B_n S_{n}}.
$$
We can deduce that $\mathcal{Z}$ is not a function of $\gamma_u$. Therefore $\mathcal{Z}$ is the chosen candidate to be $\mathcal{X}_{n+1}$. $\mathcal{H}_{n+1}$ is then true, and $\mathcal{H}_{n}$ true for all $n \geq 0$. 
Then is is easy to check that 
\beaa
\begin{aligned}
\sharpe(V_N^{\star}) &=\frac{\E_{\Prob}\left[S_N^{0}\left(V_0+\sum_{n=0}^{N}(\alpha_{n+1}^{\star})^T\tilde{S}_n-\sum_{i=1}^d\nu|\alpha_{n+1}^{i\star}-\alpha_{n}^{i\star}|\tilde{S}_n^i\right)\right]-V_0S_N^0}{\Var\left[S_N^{0}\left(V_0+\sum_{n=0}^{N}(\alpha_{n+1}^{\star})^T\tilde{S}_n-\sum_{i=1}^d\nu|\alpha_{n+1}^{i\star}-\alpha_{n}^{i\star}|\tilde{S}_n^i\right)\right]^{\frac{1}{2}}}\\
&=\frac{\E_{\Prob}\left[\left(\sum_{n=0}^{N}(\mathcal{X}_{n+1})^T\tilde{S}_n-\sum_{i=1}^d\nu|\mathcal{X}_{n+1}^{i}-\mathcal{X}_{n}^{i}|\tilde{S}_n^i\right)\right]}{\Var\left[\left(\sum_{n=0}^{N}\mathcal{X}_{n+1}^T\tilde{S}_n-\sum_{i=1}^d\nu|\mathcal{X}_{n+1}^{i}-\mathcal{X}_{n}^{i}|\tilde{S}_n^i\right)\right]^{\frac{1}{2}}},
\end{aligned}
\eeaa
which is not a function of $\gamma_u$.
\end{proof}

\bibliographystyle{siam}
\bibliography{references}

\end{document}